\documentclass[a4paper,UKenglish]{lipics-v2018}
\usepackage{thmtools,thm-restate}
\hideLIPIcs 

\nolinenumbers 
\usepackage{stmaryrd}
\usepackage{enumerate}
\usepackage{bbm}       
\usepackage{mathtools} 

\usepackage[outline]{contour}     
\usepackage{tikz}
\usetikzlibrary{patterns}
\usetikzlibrary{shapes.geometric}

\usepackage{diego-macro}

\renewcommand{\epsilon}{\varepsilon}
\newcommand{\ignore}[1]{}

\newcommand{\cost}{\textit{cost}}

\usepackage[backgroundcolor=orange!50, textsize=small, disable]{todonotes}
\usepackage{comment}

\newcommand{\sidediego}[1]{\todo[backgroundcolor=blue!20]{{\bf Diego} #1}}

\newcommand{\cent}{\text{\textcent}}

\theoremstyle{plain}
\newtheorem{proposition}[theorem]{Proposition}

\newtheorem*{claim*}{Claim}
\newtheorem*{lemma*}{Lemma}

\newtheorem*{theorem*}{Main Theorem}

\title{Boundedness of Conjunctive Regular Path Queries}

\author{Pablo Barcel{\'o}}{Department of Computer Science, University of Chile \& IMFD Chile}{}{}{}
\author{Diego Figueira}{CNRS \& LaBRI, Universit{\'e} de Bordeaux, France}{}{}{}
\author{Miguel Romero}{University of Oxford, UK}{}{}{}

\Copyright{Pablo Barcel{\'o}, Diego Figueira, and Miguel Romero}

\authorrunning{P.~Barcel{\'o}, D.~Figueira, and M.~Romero} 

\subjclass{ \, }


\EventEditors{Ioannis Chatzigiannakis, Christos Kaklamanis, Daniel Marx, and Don Sannella}
\EventNoEds{4}
\EventLongTitle{45th International Colloquium on Automata, Languages, and Programming (ICALP 2018)}
\EventShortTitle{ICALP 2018}
\EventAcronym{ICALP}
\EventYear{2018}
\EventDate{July 9--13, 2018}
\EventLocation{Prague, Czech Republic}
\EventLogo{eatcs}
\SeriesVolume{80}
\ArticleNo{}

\funding{Barcel\'o is partially funded by the Millennium Institute for Foundational Research on Data and Fondecyt 1170109, and 
Figueira by ANR project DELTA, grant ANR-16-CE40-0007. This project has received funding from the European Research Council (ERC) under the European Union's Horizon 2020 research and innovation programme (grant agreement No 714532). The paper reflects only the authors' views and not the views of the ERC or the European Commission. The European Union is not liable for any use that may be made of the information contained therein.} 

\acknowledgements{We are grateful to Thomas Colcombet for helpful discussions and valuable ideas in relation to the results of Section~\ref{sec:distance}.}

\begin{document}

\maketitle
\begin{abstract}

  We study the boundedness problem for unions of conjunctive regular path queries with inverses (UC2RPQs). This is the problem of, given a UC2RPQ, checking whether it is equivalent to a union of conjunctive queries (UCQ). 
  We show the problem to be \expspace-complete, thus coinciding with the complexity of containment for UC2RPQs. As a corollary, when a UC2RPQ is bounded, it is equivalent to a UCQ of at most triple-exponential size, and in fact we show that this bound is optimal. 
  We also study better behaved classes of UC2RPQs, namely \textsl{acyclic UC2RPQs of bounded thickness}, and \textsl{strongly connected UCRPQs}, whose boundedness problem are, respectively, \pspace-complete and $\Pi_2^P$-complete.
  Most upper bounds exploit results on limitedness for distance automata, in particular extending the model with alternation and two-wayness, 
  which may be of independent interest.
  
\end{abstract}

\section{Introduction}    
	\label{sec:introduction}

Boundedness is an important property of formulas in 
logics with fixed-point features. At the intuitive level, a formula $\phi$ in any such logic is bounded 
if its {\em fixed-point depth}, \ie, the number of iterations that are needed to evaluate 
$\phi$ on a structure ${\bf A}$, is fixed (and thus it is independent of ${\bf A}$). In databases and knowledge 
representation, 
 boundedness is regarded as an interesting theoretical phenomenon with relevant practical implications \cite{Var,BHLW16}. 
In fact, while several applications in these areas require the use of recursive features, 
actual real-world systems are either not designed or not optimized to cope with the computational demands that such features impose. 
Bounded formulas, in turn, can be reformulated in non-recursive logics, such as FO, 
or even as a {\em union of conjunctive queries} (UCQ) when $\phi$ itself is positive. UCQs form the core 
of most systems for data management and ontological query answering, and, in addition, are the focus of 
advanced optimization techniques.   
It has also been experimentally verified in some contexts 
that recursive features encountered in practice are often used in a somewhat `harmless' way, and that  
many of such queries are in fact bounded \cite{LSW15}. Thus, checking if a recursive 
formula $\phi$ is bounded, and building an equivalent non-recursive formula $\phi'$ when the latter holds, are important  
optimization tasks. 

The study of boundedness for \emph{Datalog} programs, \ie, the least fixed-point extension of the class of UCQs, received a lot of attention during the late 80s and early 90s. Two seminal results established that checking boundedness is undecidable in general for Datalog \cite{GMSV93}, but becomes decidable 
for {\em monadic} Datalog, \ie, those programs  in which each intensional predicate is monadic \cite{CGKV88}. The past few years have seen a resurgence of interest in boundedness problems. This is due, in part, to the development of the theory of {\em cost automata} over trees (both finite and infinite) in a series of landmark results, in particular relating to its {\em limitedness} problem. 
In a few words, cost automata are generalizations of finite automata associating a \emph{cost} from $\mathbb{N}\cup\{\infty\}$ to every input tree (instead of simply accepting or rejecting). The limitedness problem asks, given a cost automata, whether there is a uniform bound on the cost over all (accepting) input trees.
  Some deep results establish that checking limitedness is decidable for 
  well-behaved classes of cost automata over trees  \cite{CL08,VB11,VB12,BenediktCCB15}. Remarkably, for several logics of interest the boundedness problem can be reduced to the limitedness for cost automata in such well-behaved classes. Those reductions have enabled powerful decidability results for the boundedness problem. 
  As an example, it has been shown in this way that boundedness is decidable for monadic second-order logic (MSO) over structures of bounded treewidth \cite{BOW14}, which corresponds to an extension of Courcelle's Theorem, and also for the {\em guarded negation} fragment of least fixed-point logic (LFP), even in the presence of ungarded parameters \cite{BBB16}. Cost automata have also been used to study the complexity of boundedness for guarded Datalog programs \cite{BenediktCCB15,BBLP18}. 

{\em Graph databases} is a prominent area of study within database theory, in which the use of recursive queries is crucial \cite{B13,AABHRV17}. A graph database is a finite edge-labeled directed graph. The most basic querying mechanism for graph databases corresponds to the class of {\em regular path queries} (RPQs), which check whether two nodes of the graph are connected by a path whose label belongs to a given regular language. RPQs are often extended with the ability to traverse edges in both directions, giving rise to the class of {\em two-way} RPQs, or 2RPQs \cite{CGLV02}. The core of the most popular recursive query languages for graph databases is defined by {\em conjunctive} 2RPQs, or C2RPQs, which are the closure of 2RPQs under conjunction and existential quantifications \cite{CGLV00}. We also consider {\em unions} of C2RPQs, or UC2RPQs. It can be shown that a UC2RPQ is bounded iff it is equivalent to some UCQ. In spite of the inherent recursive nature of UC2RPQs, their boundedness problem has not been studied in depth. 
Here we develop such a study by showing the following: 
\begin{itemize} 
\item The boundedness problem for UC2RPQs is \expspace-complete. The lower bound holds even for CRPQs. This implies that boundedness 
is not more difficult than {\em containment} for UC2RPQs, which was shown to be \expspace-complete in \cite{CGLV00}.  
\item From our upper bound construction it follows that if a UC2RPQ is bounded, then it is equivalent to a UCQ of triple-exponential size. We show that this bound is optimal.  
\item Finally, we obtain better complexity bounds for some subclasses of UC2RPQs; namely, for acyclic UC2RPQs of bounded \emph{thickness}, in which case boundedness becomes \pspace-complete, and for strongly connected UCRPQs, for which it is $\Pi_2^P$-complete. \sidediego{Tenemos lower bounds sobre ``size of equivalent UCQ'' en estos casos?}
\end{itemize} 

It is important to stress that UC2RPQs can be easily translated into guarded LFP with ungarded parameters, for which boundedness was shown to be decidable  by applying sophisticated cost automata techniques as mentioned above. However, the complexity of the boundedness problem for such a logic is currently not well-understood -- and it is at least {\sc 2Exptime}-hard \cite{BenediktCCB15} -- and hence this translation does not yield, in principle, optimal complexity bounds for our problem. To study the boundedness for UC2RPQs, we develop instead
 techniques especially tailored to UC2RPQs. 
In fact, since the recursive structure of UC2RPQs is quite tame, their boundedness problem can be translated 
into the  limitedness problem for a much simpler automata model than cost automata on trees; namely, {\em distance} automata on finite words. 
 Distance automata are nothing more than usual NFAs with two sorts of transitions: costly and non-costly. Such an automaton is limited if there is an integer $k \geq 1$ such that every word accepted by the NFA has an accepting run with at most $k$ costly transitions.
A beautiful result in automata theory established the decidability of the limitedness problem for distance automata \cite{hashiguchi1982limitedness}, actually in \pspace \cite{LeungP04}. While this continues being a difficult result, by now we have quite transparent proofs of this fact (see, \eg, \cite{Kirsten05}). 
We exploit our translation to obtain tight complexity upper bounds for boundedness of UC2RPQs. Some of the proofs in the paper require extending the study of limitedness to {\em alternating} and {\em two-way} distance automata, while preserving the \pspace bound for the limitedness problem. We believe these results to be of independent interest.

\section{Preliminaries}   
	\label{sec:preliminaries}

We assume familiarity with {\em non-deterministic finite automata} (NFA), 
{\em two-way} NFA (2NFA), and {\em alternating finite automata} (AFA) over finite words. We often blur the distinction between an NFA $\+A$ and the language $L(\+A)$ it defines; similarly for regular expressions. 

\smallskip 
\noindent
 {\bf {\em Graph databases and conjunctive regular path queries.}} 
A \defstyle{graph database} over a finite alphabet $\A$ is a finite edge-labelled graph $G = (V, E)$ over $\A$, where $V$ is a finite set of vertices and $E \subseteq V \times \A \times V$ is the set of labelled edges. We write $u \xrightarrow{a} v$ to denote an edge $(u,a,v) \in E$. \sidediego{no usamos también $u \xrightarrow{w} v$ para palabras $w$?}
 We define the alphabet $\A^\pm:= \A\dcup \A^{-1}$ that extends $\A$ with the set $\A^{-1} := \{a^{-1} \mid a \in \A\}$ of  ``inverses'' of symbols in $\A$. 
An {\em oriented path} from $u$ to $v$ in a graph database $G=(V,E)$ over alphabet $\A$ is a pair $\pi=(\sigma,\ell)$ where $\sigma$ and $\ell$ are (possibly empty) sequences $\sigma = (v_0,a_1,v_1),(v_1,a_2,v_2),\dots,(v_{k-1},a_k,v_k)\in V\times\A\times V$, and $\ell=\ell_1,\dots,\ell_k \in \{-1,1\}$, for $k \geq 0$, such that $u=v_0$, $v=v_k$, and for each $1 \leq i \leq k$, we have that 
$\ell_i=1$ implies $(v_{i-1},a_i,v_i) \in E$; and $\ell_i=-1$ implies $(v_{i},a_i,v_{i-1}) \in E$. 
The {\em label} of $\pi$ is the word $b_1 \dots b_k \in (\A^\pm)^*$, where $b_i=a_i$ if $\ell_i=1$; otherwise $b_i=a_i^{-1}$. 
When $k = 0$ the label of $\pi$ is the empty word $\epsilon$. 
If $\ell_i=1$ for every $1\leq i\leq k$, we say that $\pi$ is a \emph{directed path}.  
Note that in this case, the label of $\pi$ belongs to $\A^*$.

A \defstyle{regular path query} (RPQ)  over   $\A$ 
is a regular language $L \subseteq \A^*$, 
which we assume given as an NFA. The evaluation of $L$ on a graph database $G = (V,E)$ over $\A$, written $L(G)$, 
is the set of pairs $(u,v)\in V\times V$ such that there is a directed path from $u$ to $v$  
 in $G$ whose label belongs to $L$.
 2RPQs extend RPQs with the ability to traverse edges in both directions. Formally, a 2RPQ $L$ over 
 $\A$ is simply an RPQ over  $\A^\pm$. 
 The evaluation $L(G)$ of $L$ over a graph database $G = (V,E)$ over $\A$ is 
the set of pairs $(u,v)\in V \times V$ such that there is an oriented path from $u$ to $v$ in $G$ whose label belongs to $L$. 

 \defstyle{Conjunctive 2RPQs} (C2RPQs) are obtained by taking  
 the closure of 2RPQs under conjunction and existential quantification, \ie, 
 a C2RPQ over $\A$ is an expression $\gamma := \exists \bar z \,\big((x_1 \xrightarrow{L_1} y_1) \wedge \dots \wedge (x_m \xrightarrow{L_m} y_m)\big)$, 
 where each $L_i$ is a 2RPQ over $\A$ and $\bar z$ is a tuple of variables among those in 
 $\{x_1,y_1,\dots,x_m,y_m\}$. We say that $\gamma$ is a CRPQ if each $L_i$ is an RPQ. 
 If $\bar x = (x_1,\dots,x_n)$ is the tuple of {\em free variables} of $\gamma$, \ie, those that 
 are not existentially quantified in $\bar z$, then the evaluation $\gamma(G)$ of the C2RPQ $\gamma$ over a graph database $G$
is the set of all tuples $h(\bar x) = (h(x_1)\dots,h(x_n))$,
 where $h$ ranges over all mappings $h : \{x_1,y_1,\dots,x_m,y_m\} \to V$ such that $(h(x_i),h(y_i)) \in L_i(G)$ for each $1 \leq i \leq m$.

A {\em union} of C2RPQs (UC2RPQ) is an expression of the form $\Gamma := \bigvee_{1 \leq i \leq n} \gamma_i$, where the $\gamma_i$'s are C2RPQ, all of which have exactly the same free variables. 
The evaluation $\Gamma(G)$ 
of $\Gamma$ over a graph database $G$ is $\bigcup_{1 \leq i \leq n} \gamma_i(G)$. We often write $\Gamma(\bar x)$ to denote that 
$\bar x$ is the tuple of free variables of $\Gamma$. 
A UC2RPQ $\Gamma$ is {\em Boolean} if it contains no free variables. 

Given UC2RPQs $\Gamma$ and $\Gamma'$, we write $\Gamma \subseteq \Gamma'$ if $\Gamma(G) \subseteq \Gamma'(G)$ for each graph database 
$G$. Hence, $\Gamma$ and $\Gamma'$ are {\em equivalent} if $\Gamma \subseteq \Gamma'$ and $\Gamma' \subseteq \Gamma$, \ie, $\Gamma(G) = \Gamma'(G)$ for every $G$. 
%

\smallskip 
\noindent
 {\bf {\em Boundedness of UC2RPQs.}} 
 CRPQs, and even UC2RPQs, can easily be expressed in \emph{Datalog}, the least fixed-point extension of the class of {\em union of conjunctive queries} (UCQs). 
 Hence, we can directly define the boundedness of a UC2RPQ in terms of the boundedness of its equivalent Datalog program, which is a well-studied problem \cite{Var}. The latter, however, 
 coincides with being equivalent to some UCQ \cite{Nau}.
In the setting of graph databases, a {\em conjunctive query} (CQ) over $\A$ is simply a CRPQ over $\A$ of the form $\exists \bar z \bigwedge_{1 \leq i \leq m} 
(x_i \xrightarrow{a_i} y_i)$ where the $a_i$s range over $\A \cup \set{\epsilon}$. 
Notice that atoms of the form $x \xrightarrow{\epsilon} y$ correspond to \emph{equality atoms} $x=y$. Analogously, one can define unions of CQs (UCQs). Note that, modulo equality atoms, a CQ over $\A$ can be seen as a graph database over $\A$. Hence, we shall slightly abuse notation and use in the setting of CQs, notions defined for graph databases (such as oriented paths). 



 
 
 %
 A UC2RPQ $\Gamma$ is {\em bounded} if it is equivalent to some UCQ $\Phi$. In this article we study the complexity of the problem
  {\sc Boundedness}, which takes as input a UC2RPQ $\Gamma$ and asks whether $\Gamma$ is bounded.  
 
  \begin{example} 
 \label{ex:example1}\sidediego{note to self: check}
Consider the Boolean UCRPQ $\Gamma=\gamma_1\lor \gamma_2$ over the alphabet $\A=\{a,b,c,d\}$ such that 
$\gamma_1=\exists x,y\, (x \xrightarrow{L_b} y \land x \xrightarrow{L_{b,d}} y)$ and $\gamma_2=\exists x,y\, (x \xrightarrow{L_d} y\land x \xrightarrow{L_{b,d}} y)$, 
where $L_b:=a^+b^+c$, $L_d:=ad^+c^+$, and $L_{b,d}:=a^+(b+d)c^+$. 
For $e\in \A$, recall that $e^+$ denotes the language $e(e^*)$. 
As we shall explain in Example~\ref{ex:example2}, we have that $\gamma_1$ and $\gamma_2$ are unbounded. 
However, $\Gamma$ is bounded, and in particular, it is equivalent to the UCQ $\Phi=\phi_1\lor \phi_2$, where 
$\phi_1$ and $\phi_2$ correspond to $\exists x,y\, (x \xrightarrow{abc} y)$ and $\exists x,y\, (x \xrightarrow{adc} y)$, respectively.\qed
  \end{example}

\medskip
\noindent
{\bf {\em Organization of the paper.}} We present characterizations of boundedness for UC2RPQs in Section \ref{sec:char} 
and an application of those to pinpoint the complexity of  {\sc Boundedness} for RPQs in Section \ref{sec:sec-rpq-bound}. 
Distance automata and results about them are given in Section \ref{sec:distance}. We analyze the complexity of 
{\sc Boundedness} for general UC2RPQs in Section \ref{sec:c2rpq-bound} and present some classes of UC2RPQs with 
better complexity of {\sc Boundedness} in Section \ref{sec:nice-c2rpq}. We finish with a discussion 
in Section \ref{sec:discussion}. Due to space constraints many proofs are in the appendix. 

 \section{Characterizations of Boundedness for UC2RPQs}
 \label{sec:char} 
 
 In this section we provide two simple characterizations of when a UC2RPQ is bounded that will be useful to analyze the complexity of {\sc Boundedness}. 
Let $\phi(\bar x)$ and $\phi'(\bar x)$ be CQs over $\A$ with variable sets ${\cal V}$ and ${\cal V}'$, respectively. 
Let $=_\phi$ and $=_{\phi'}$ be the binary relations induced on ${\cal V}$ and ${\cal V}'$ by the equality atoms of 
$\phi$ and $\phi'$, respectively, and $=_\phi^*$ and $=_{\phi'}^*$ be their reflexive-transitive closure. 
A {\em homomorphism} from $\phi$ to $\phi'$ is a mapping 
$h : {\cal V} \to {\cal V'}$ such that: (i) $x =_\phi^* y$ implies $h(x)=_{\phi'}^*h(y)$; (ii) $h(\bar x) = \bar x$; and 
(iii) for each atom $x \xrightarrow{a} y$ in $\phi$ with $a \in \A$, there is an atom $x' \xrightarrow{a} y'$ in $\phi'$ such that $h(x) =_{\phi'}^* x'$ and $h(y) =_{\phi'}^* y'$. 
We write $\phi \to \phi'$ if such a homomorphism exists. 
It is known that $\phi \to \phi'$ iff $\phi' \subseteq \phi$ \cite{CM77}. 

An \defstyle{expansion} of a C2RPQ $\gamma(\bar x)$ over $\A$ is a CQ $\lambda(\bar x)$ over $\A$ with minimal number of variables and atoms such that 
(i) $\lambda$ contains each variable of $\gamma$, (ii) for each atom $A=x \xrightarrow{L} y$, there is an oriented path $\pi_A$ in $\lambda$ from $x$ to $y$ with label $w_A\in L$ 
whose intermediate variables (\ie, those not in $\{x,y\}$) are distinct from one another, and (iii) intermediate variables of different oriented paths $\pi_A$ and $\pi_{A'}$ are disjoint. 
Note that the free variables of $\lambda$ and $\gamma$ coincide. 
Intuitively, the expansion $\lambda$ is obtained from $\gamma$ by choosing for each atom $A=x \xrightarrow{L} y$ a word $w_A\in L$, 
and ``expanding'' $x \xrightarrow{L} y$ into the ``fresh oriented path'' $\pi_A$ from $x$  to $y$ with label $w_A$.
When $w_A=\epsilon$ then $\lambda$ contains the equality atom $x=y$. 
An expansion of a UC2RPQ $\Gamma$ is an expansion of some C2RPQ in $\Gamma$. 
Observe that a (U)C2RPQ is always equivalent to the (potentially infinite) UCQ given by its set of expansions. 
Even more, it is equivalent to the UCQ defined by its {\em minimal} expansions, as introduced below.

If $\lambda$ is an expansion of a UC2RPQ $\Gamma$, 
we define the \emph{size} of $\lambda$, denoted by $\|\lambda\|$, to be the number of (non-equality) atoms in $\lambda$. 
We say that $\lambda$ is {\em minimal}, if there is no
expansion $\lambda'$ such that $\lambda' \to \lambda$ and $\|\lambda'\| < \|\lambda\|$.  
Intuitively, an expansion is minimal if its answers cannot be covered by a smaller expansion.  
We can then establish the following. 

\begin{lemma} \label{lemma:equiv-minimal}
Every UC2RPQ $\Gamma$ is equivalent to the (potentially infinite) UCQ given by its set of minimal expansions. 
\end{lemma} 

We can now provide our basic characterizations of boundedness.  

\begin{proposition} \label{prop:basic} 
 The following conditions are equivalent for each UC2RPQ $\Gamma$. 
 \begin{enumerate} 
 \item $\Gamma$ is bounded.
 \item There is $k \geq 1$ such that for every expansion $\lambda$ of $\Gamma$ 
  there exists an expansion $\lambda'$ of $\Gamma$ with $\|\lambda'\| \leq k$ such that 
  $\lambda \subseteq \lambda'$ (\ie, such that $\lambda' \to \lambda$).
  \item $\Gamma$ has finitely many minimal expansions. \sidediego{acá tenemos que trabajar modulo ``renaming of variables'', tal como definimos las expansiones siempre hay infinitas...}
\end{enumerate} 
\end{proposition}

\begin{example}
\label{ex:example2}
Consider the Boolean UCRPQ $\Gamma=\gamma_1\lor\gamma_2$ over $\A=\{a,b,c,d\}$ from Example~\ref{ex:example1}. 
To see that $\gamma_1$ is unbounded (the case of $\gamma_2$ is similar) we can apply Proposition \ref{prop:basic}. 
Indeed, the expansions of $\gamma_1$ corresponding to $\{\exists x,y\, (x \xrightarrow{ab^nc} y \land x \xrightarrow{adc} y): n\geq 1\}$ 
are all minimal. On the other hand, $\Gamma$ is bounded as its minimal expansions correspond to $\exists x,y\, (x \xrightarrow{abc} y \land x \xrightarrow{abc} y)$ and 
$\exists x,y\, (x \xrightarrow{adc} y \land x \xrightarrow{adc} y)$.\qed
\end{example}

\section{Boundedness for Existentially Quantified RPQs}
	\label{sec:sec-rpq-bound}

As a first application of Proposition \ref{prop:basic}, we study {\sc Boundedness} for CRPQs consisting of a single RPQ; that is, 
RPQs or existentially quantified RPQs. 
Let $v,w$ be words over $\A$. 
Recall that a word $v$ is a \defstyle{prefix} \resp{\defstyle{suffix} and \defstyle{factor}}
of $w$ if $w\in v\cdot \A^*$ \resp{$w\in\A^*\cdot v$ and $w\in\A^*\cdot v \cdot \A^*$}. If in addition we have $v\neq w$, 
then we say that $v$ is a \emph{proper} prefix \resp{suffix and factor} of $w$. 
For a language $L\subseteq \A^*$, we define its \emph{prefix-free} sub-language $L_{\text{pf}}$ to be the set of words $w\in L$ such that 
$w$ has no proper prefix in $L$. Similarly, we define $L_{\text{sf}}$ and $L_{\text{ff}}$ with respect to the suffix and factor relation. 
We have the following:

\begin{proposition} \label{prop:basic-rpq} 
The following statements hold. 
	\begin{enumerate}
		\item An RPQ $L$ is bounded iff $L$ is finite. 
		\item A CRPQ $\exists y (x \xrightarrow{L} y)$ \resp{$\exists x (x \xrightarrow{L} y)$} with $x\neq y$ is bounded iff $L_{\text{pf}}$ \resp{$L_{\text{sf}}$} is finite. 
		\item A Boolean CRPQ $\exists x,y ( x \xrightarrow{L} y)$ with $x\neq y$ is bounded iff $L_{\text{ff}}$ is finite. 
	\end{enumerate}
\end{proposition}

 \begin{theorem} \label{theo:finite-lang} 
The problem of, given an NFA accepting the language $L$, checking whether $L_{\text{pf}}$ is finite is \pspace-complete. 
The same holds if we replace $L_{\text{pf}}$ by $L_{\text{sf}}$ or $L_{\text{ff}}$. 
\end{theorem}
 
 \begin{proof} 
 We focus on upper bounds, the lower bounds are in the appendix. 
Given an NFA $\cal A$ accepting the language $L$, we can construct an NFA ${\cal B}$ of polynomial size in $\cal A$ that accepts precisely those words that have a proper prefix in $L$. 
By complementing and intersecting with $\cal A$, we obtain an NFA ${\cal B}'$ of exponential size in $\cal A$ that accepts the language $L_{\text{pf}}$. 
Hence, we only need to check whether the language accepted by ${\cal B}'$ is finite, which can be done \textsl{on-the-fly} in \nlogspace w.r.t. $\cal B'$, 
and hence in \pspace. 
 The other two cases are analogous.  
 \end{proof} 
 
 By applying Theorem \ref{theo:finite-lang} and Proposition \ref{prop:basic-rpq}, we can now pinpoint the complexity of {\sc Boundedness} for CRPQs with a single RPQ. 
 
 \begin{corollary} \label{coro:comp-rpq} 
The following statements hold. 
\begin{enumerate} 
\item {\sc Boundedness} for RPQs is \nlogspace-complete. 
\item {\sc Boundedness} for CRPQs of the form $\exists y (x \xrightarrow{L} y)$, with $x\neq y$, is \pspace-complete. The same holds for CRPQs 
$\exists x (x \xrightarrow{L} y)$ and Boolean CRPQs $\exists x,y (x \xrightarrow{L} y)$, where $x\neq y$. 
\end{enumerate} 
\end{corollary}

It is not clear, though, how usual automata techniques, as the ones applied in the proof of
Theorem \ref{theo:finite-lang}, can be used to solve {\sc Boundedness} for 
more complex CRPQs. 
To solve this problem 
we develop an approach based on distance automata, as introduced next. Our 
approach also handles inverses and unions, thus 
dealing with arbitrary UC2RPQs.

\section{Distance Automata}   
	\label{sec:distance}

Distance automata \cite{hashiguchi1982limitedness} (equivalent to weighted automata over the $(\min,+)$-semiring  \cite{droste2009handbook}, min-automata \cite{BojanczykT09}, or $\set{\epsilon,ic}$-B-automata \cite{Colcombet09}) are 
an extension of finite automata which associate to each word in the language a natural number or `cost'. They can be represented as non-deterministic finite automata with two sorts of transitions: costly and non-costly. For a given distance automaton, the cost of a run on a word is the number of costly transitions, and the cost of a word $w \in \A^*$ is the minimum cost of an accepting run on $w$. We will use this automaton model to encode boundedness as the problem of whether there is a uniform bound on the cost of words, known as the \textsl{limitedness problem}.

Formally, a \defstyle{distance automaton} (henceforth \defstyle{DA}) is a tuple $\+A = (\A, Q, q_0, F, \delta)$, where $\A$ is a finite alphabet, $Q$ is a finite set of states, $q_0 \in Q$ is the initial state, $F \subseteq Q$ is the set of finals states and $\delta \subseteq Q \times \A \times \set{0,1} \times Q$ is the transition relation. A word $w \in \A^*$ is {\em accepted} by $\+A$ if there is an {\em accepting run} of $\+A$ on $w$, \ie, a (possibly empty) sequence of transitions $\rho = (p_1,a_1,c_1,r_1) \dotsb (p_n,a_n,c_n,r_n) \in \delta^*$ with the usual properties: (1) if $\rho = \epsilon$ then $q_0 \in F$ and $w = \epsilon$, (2) $p_1 = q_0$ and $r_n \in F$, (3) for every $1\leq i<n$ we have $r_i=p_{i+1}$, and (4) $w = a_1 \dotsb a_n$. The {\em cost} of the run $\rho$ is $\cost(\rho)=c_1 + \dotsb + c_n$ (or $0$ if $\rho=\epsilon$); and the cost \defstyle{$\cost_{\+A}(w)$} of a word 
$w$ accepted by $\+A$ is the minimum cost of an accepting run of $\+A$ on $w$. For convenience, 
we assume the cost of words not accepted by $\+A$ to be $0$.

The \defstyle{limitedness problem} for DA is defined as follows: given a DA $\+A$, determine 
whether  $\sup_{w \in \A^*} \cost_{\+A}(w) < \infty$. This problem is known to be \pspace-complete.

\begin{theorem}\label{thm:da-pspace}\cite{Leung91,LeungP04}
The following statements hold: 
  \begin{enumerate}
  \item  The limitedness problem for DA is
    \pspace-complete.
  \item If a DA with $n$ states is limited, then $\sup_{w \in \A^*} \cost_{\+A}(w) \leq 2^{O(n^3)}$.
\end{enumerate}
\end{theorem}

We use two extensions of DA: \emph{alternating} and \emph{two-way}.  Two-way DA is defined as for NFA, extending the cost function accordingly. The cost of a word is still the minimum over the cost of all (potentially infinitely many) runs. Alternating DA is defined as usual by having two sorts of states: universal and existential. Existential states can be seen as computing the minimum among the cost of all possible continuations of the run, and universal states as computing the maximum (or supremum if the automaton is also two-way). As we will see, these extensions preserve the above \pspace upper bound for the limitedness problem.

\newcommand{\ldel}{{\vdash}}%
\newcommand{\rdel}{{\dashv}}%
\newcommand{\yend}{\textit{end}}%
\newcommand{\nend}{\overline{\textit{end}}}%
Formally, an \defstyle{alternating two-way DA with epsilon transitions} (A2DA$^\epsilon$) over $\A$ 
is a tuple $\+A = (\A, Q_\exists, Q_\forall, q_0, F, \delta)$ is a A2DA$^{\epsilon}$ if $q_0 \in Q_\exists$, $F \subseteq Q_\exists$ and 
\[
\delta \subseteq (Q_\exists \cup Q_\forall) \times (\A^\pm \cup \set\epsilon) \times \set{\yend,\nend} \times \set{0,1} \times (Q_\exists \cup Q_\forall);
\]
where $\yend$ indicates that after reading the letter we arrive at the end of the word (\ie, either the leftmost or the rightmost end) and 
$\nend$ indicates that we do not. 
When the automaton $\+A$ is two-way, it is convenient to think of its head as being \emph{between} the letter positions of the word, so an $\yend$-flagged transition can be applied only if it moves the head to be right before the first letter of the word, or right after the last one.

For any given word $w \in \A^*$, consider the edge-labelled graph $G_{\+A,w}=(V,E)$ over $\delta$, where $V = Q \times \set{0, \dotsc, |w|}$, with $Q=Q_\exists \cup Q_\forall$, and $E \subseteq V \times \delta \times V$ consists of all edges $(q,i) \xrightarrow{(q,a,e,c,p)} (p,j)$ such that $e = \yend$ iff $j = 0$ or 
$j = |w|$ and either (a) $i<|w|$, $a = w[i+1]$, and $j = i+1$; (b) $i>0$, $a = (w[i])^{-1}$, and $j=i-1$; or (c)  $a = \epsilon$ and $j = i$.

An \defstyle{accepting run of $\+A$ on $w$ from $(q,i) \in Q \times \set{0, \dotsc, |w|}$} is a finite (possibly empty) edge-labelled directed rooted tree\footnote{That is, a tree-shaped finite edge-labelled graph over $\delta$ with edges directed in the root-to-leaf sense.} 
$t$ over $\delta$ and a labelling $h$ from the nodes of $t$ to the nodes of $G_{\+A,w}$, such that if $t$ is empty then $q \in F$, 
and otherwise $h$ maps the root of $t$ to 
$(q,i)$, every leaf of $t$ to $F \times \set{0, \dotsc, |w|}$, and
for every node $x$ of $t$: 
\begin{itemize}
\item if $(x,\alpha,y)$ is an (labeled) edge in $t$ for some $y$, then $(h(x),\alpha,h(y))$ is an edge in $G_{\+A,w}$;
	\item if $h(x) \in Q_\forall \times \set{0, \dotsc, |w|}$, then for every edge $(h(x),\alpha,c)$ in $G_{\+A,w}$, there is an edge $(x,\alpha,y)$ in $t$ so that $h(y) = c$;
	\item if $h(x) \in Q_\exists \times \set{0, \dotsc, |w|}$, then $x$ has at most one child.
\end{itemize}

Each branch of $t$ with label $(q_1, a_1,e_1,c_1,p_1), \dotsc, (q_n, a_n,e_n,c_n,p_n)$ has an associated cost of $c_1 + \dotsb + c_n$; and the cost associated with $t$ is the maximum among the costs of its branches, or $0$ if $t$ is empty. 
The cost $\cost_{\+A}(w,q,i)$ is the minimum cost of an accepting run on $w$ from $(q,i)$, or $0$ if none exists; $\cost_{\+A}(w)$ is defined as $\cost_{\+A}(w,q_0,0)$.

An A2DA$^{\epsilon}$ with $\delta \subseteq Q \times (\A \cup \set\epsilon) \times \set{\yend,\nend} \times \set{0,1} \times Q$ is an \defstyle{alternating DA} with $\epsilon$ transitions (ADA$^\epsilon$). An A2DA$^\epsilon$ with $Q_\forall = \emptyset$ is a \defstyle{two-way DA} with $\epsilon$ transitions (2DA$^\epsilon$). An A2DA with both the aforementioned conditions is (equivalent to) a DA with $\epsilon$ transitions (DA$^\epsilon$). Notice that in the last two cases, accepting runs can be represented as words from $\delta^*$ rather than trees. By A2DA (resp., ADA, 2DA, DA) we denote a A2DA$^\epsilon$ (resp., ADA$^\epsilon$, 2DA$^\epsilon$, DA$^\epsilon$) with no $\epsilon$-transitions. Note that DA as just defined is in every sense equivalent to the distance automata model we have defined at the beginning of this section ---this is why we overload the same `DA' name.

\smallskip

We first observe that 2DA can be transformed into DA while preserving both the language and limitedness problems by adapting the standard ``crossing sequence'' construction for translating 2NFA into NFA \cite{shepherdson1959reduction}. This fact will be useful for proving the \expspace upper bound for {\sc Boundedness} of general UC2RPQs in Section~\ref{sec:c2rpq-bound}.

\begin{proposition}\label{prop:2da-da}
There is an exponential time procedure which for every 2DA $\+A$ over $\A$ produces a  
DA $\+B$ over $\A$ such that the languages accepted by $\+A$ and $\+B$ are the same, and 
$\cost_{\+B}(w) \leq \cost_{\+A}(w) \leq f(\cost_{\+B}(w))$ for every $w \in \A^*$, where $f$ is a polynomial function that depends on the
 statespace of $\+A$.
\end{proposition}


The universality problem for NFAs is known to be \pspace-complete \cite{Kozen77}. The upper bound actually extends to 2NFA and 
even AFA. 
We show that, likewise, the limitedness problem remains in \pspace for A2DA$^\epsilon$. This result will be useful to show in Section~\ref{sec:nice-c2rpq} that {\sc Boundedness} for the class of acyclic UC2RPQs of bounded thickness is in \pspace.

\begin{theorem}\label{thm:a2da-limitedness-pspace}
  The limitedness problem for A2DA$^\epsilon$ is \pspace-complete.
\end{theorem}

The novelty of this result is the \pspace upper bound. 
In fact, decidability follows from known results, and in particular \cite[Theorem~14]{BenediktCCB15} claims \exptime-membership in the 
more challenging setup of infinite trees. However, this is obtained via an involved construction spanning through several papers.
The proof of Theorem~\ref{thm:a2da-limitedness-pspace}, instead, is obtained by the composition of the following reductions:
  \[
  \text{lim.~A2DA$^\epsilon$} \xrightarrow{(1)} \text{lim.~A2DA} \xrightarrow{(2)}  \text{lim.~2DA} \xrightarrow{(3)} \text{lim.~ADA$^\epsilon$} \xrightarrow{(4)} \text{lim.~ADA} \xrightarrow{(5)} \text{lim.~DA.}
  \]
Reductions (1), (3) and (4) are in polynomial time, while reductions (2) and (5), which are basically the same, are in exponential time. Specifically, reductions (2) and (5) preserve the statespace but the size of the alphabet grows exponentially in the number of states and linearly in the size of the source alphabet. However, the alphabet and transition set resulting from these reductions can be succinctly described: letters are encoded in polynomial space, and checking for membership in the transition set is polynomial time computable. 

In summary, the composition (1)+(2)+(3)+(4)+(5) yields a DA with the following characteristics: (i) it has a polynomial number of states $Q$; (ii) it runs on an exponential alphabet $\A$ ---and every letter is encoded in polynomial space---; and (iii) one can check in polynomial time whether a tuple $t\in Q \times \A \times \set{\yend,\nend} \times \set{0,1} \times Q$ is in its transition relation.
This, coupled with Theorem~\ref{thm:da-pspace}, item (2) (which offers a bound depending only on the number of states), provides a polynomial space algorithm for the limitedness of A2DA$^\epsilon$: We can non-deterministically check the existence of a word with cost greater than the singly-exponential bound $N$ using only polynomial space, by guessing one letter at a time and keeping the set of reachable states together with the associated costs, where each cost is encoded in binary using polynomial space if it is smaller than $N$, or with a `$\infty$' flag otherwise. 
The algorithm accepts if at least one final state is reached and the costs of all reachable final states are marked $\infty$. 
Since \npspace=\pspace (Savitch's Theorem), Theorem~\ref{thm:a2da-limitedness-pspace} follows.

We now provide a brief description of the reductions used in the proof of Theorem \ref{thm:a2da-limitedness-pspace}. 
\begin{description}
	\item[(1) From A2DA$^\epsilon$ to A2DA] This is a trivial reduction obtained by simulating $\epsilon$-transitions by reading $a \cdot a^{-1}$ for some $a \in \A$.
	\item[(2) From A2DA to 2DA] 
	Given a A2DA $\+A = (\A, Q_\forall, Q_\exists, q_0, F, \delta)$, we build a 2DA $\+B$ over a larger alphabet $\B$, where we trade alternation for extra alphabet letters. The alphabet $\B$ consists of triples $(f^\rightarrow,a,f^\leftarrow)$, where $a \in \A$ and $f^\rightarrow,f^\leftarrow : Q_\forall \to \delta$. The idea is that $f^\rightarrow,f^\leftarrow$ are ``choice functions'' for the alternation: whenever we are to the left (resp., right) of a position of the word labelled $(f^\rightarrow,a,f^\leftarrow)$ in state $q \in Q_\forall$, instead of exploring all transitions departing from $q$ and taking the maximum cost over all such runs (this is what alternation does in $\+A$), $\+B$ chooses to just take the transition $f^\rightarrow(q)$ (resp., $f^\leftarrow(q)$). Note that $\B$ is exponential in the number of states but not in the size of $\A$. In this way, we build a 2DA $\+B$ having the same set of states as $\+A$ but with a transition function which is essentially deterministic on the states of $Q_\forall$. In the end we obtain that
	\begin{itemize}
		\item for every $w \in \B^*$, $\cost_{\+B}(w) \leq \cost_{\+A}(w_\A)$; and
		\item for every $w\in \A^*$ there is $\widetilde w\in \B^*$ so that $\widetilde w_\A=w$ and $\cost_{\+A}(w) = \cost_{\+B}(\widetilde w)$,
	\end{itemize}
	where $w_\A$ and $\widetilde w_\A$ denote the projections onto the alphabet $\A$. This implies that the limitedness problem is preserved.

	\item[(3)+(4) From 2DA to ADA]
	We show a polynomial-time translation from 2DA to ADA which preserves limitedness. In the case of finite automata, there are language-preserving reductions from 2NFA to AFA with a quadratic blowup in the statespace \cite{birget1993state,PitermanV03}. However, these translations, when applied blindly to reduce from 2DA to ADA, preserve neither the cost semantics nor the limitedness of languages. On the other hand, \cite{BlumensathCKPB14} shows an involved construction that results in a reduction from 2DA to ADA on \textsl{infinite trees}, which preserves limitedness but it is not polynomial in the number of states. We show a translation from 2DA to ADA which serves our purpose: it preserves limitedness and it is polynomial time computable.
	The translation is close to the language-preserving reduction from 2NFA to AFA of \cite{PitermanV03}, upgraded to take into account the cost of different alternation branches, somewhat in the same spirit as the \textsl{history summaries} from \cite{BlumensathCKPB14}. 

\item[(5) From ADA to DA] This is exactly the same reduction as (1), noticing that the alphabet will still be single-exponential in the original A2DA$^\epsilon$.
	
\end{description}


\section{Complexity of Boundedness for UC2RPQs}   
	\label{sec:c2rpq-bound}

Here we show that {\sc Boundedness} for UC2RPQs is \expspace-complete. We do so 
by applying distance automata results presented in the previous section on top of the semantic characterizations presented in Section \ref{sec:char}. The lower bound applies even for 
CRPQs. We further show that there is a triply exponential tight bound for the size of the equivalent UCQ of a UC2RPQ (and even CRPQ), 
whenever this exists. This is summarized in the following theorem. 
If $\Gamma$ is a UC2RPQ, we write $\|\Gamma\|$ for the length of an arbitrary reasonable encoding of $\Gamma$ ---in particular, encodings in which regular languages are described through NFA or regular expressions. 

\begin{theorem} \label{theo:main}
The following statements hold. 
\begin{enumerate}
\item {\sc Boundedness} for UC2RPQs is \expspace-complete. The problem remains \expspace-hard even for Boolean CRPQs.
\item If a UC2RPQ $\Gamma$ is bounded, there is a UCQ $\Phi$ that is equivalent to $\Gamma$ and such that $\Phi$ has at 
most triple-exponentially many CQs, each one of which is at most of double-exponential size with respect to $\|\Gamma\|$.
\item There is a family $\{\Gamma_n\}_{n \geq 1}$ of Boolean CRPQs such that for each $n\geq 1$ it is the case that: 
(1) $\|\Gamma_n\|=O(n)$, (2) $\Gamma_n$ is bounded, and (3) every UCQ that is equivalent to $\Gamma_n$ has at least triple-exponentially many 
CQs with respect to $n$. 
\end{enumerate} 
\end{theorem}

\subsection{Upper bounds}
\label{sec:upper-bound}
Our upper bound proof builds on top of techniques developed by Calvanese et al.~\cite{CGLV00} for studying the {\em containment problem for UC2RPQs}: 
Given UC2RPQs $\Gamma,\Gamma'$, is it the case that $\Gamma \subseteq \Gamma'$? 
It is shown in \cite{CGLV00} that from $\Gamma,\Gamma'$ it is possible to construct exponentially sized 
NFAs $\+A_{\Gamma,\Gamma'}$ and $\+A'_{\Gamma,\Gamma'}$, such that $\Gamma \subseteq \Gamma'$ iff there is a word in 
$\+A_{\Gamma,\Gamma'} \cap \overline{\+A'_{\Gamma,\Gamma'}}$. 
It is a well-known result that the latter is solvable in 
\nlogspace on the combined size of $(\+A_{\Gamma,\Gamma'},\overline{\+A'_{\Gamma,\Gamma'}})$, \ie, in \expspace. 
We modify this construction to study the boundedness of a given UC2RPQ $\Gamma$. In particular, we  
construct from $\Gamma$ in exponential time a DA $\+D_\Gamma$ such that $\Gamma$ is bounded iff $\+D_\Gamma$ is limited. 
The result then follows from Theorem \ref{thm:da-pspace}, 
which establishes that limitedness for $\+D_\Gamma$ can be solved in polynomial space 
on the number of its states, and thus in \expspace. 

\begin{proposition} \label{prop:bounded-2da}
There is a single-exponential time procedure that takes as input a UC2RPQ $\Gamma$ and constructs a DA $\+D_\Gamma$ such that 
$\Gamma$ is bounded iff $\+D_\Gamma$ is limited. 
\end{proposition}  

\begin{proof}
Similarly as done in \cite{CGLV00}, the DA $\+D_\Gamma$ will run over encodings of expansions of the UC2RPQ $\Gamma$, \ie, words
over the alphabet $\A_1 := \A^{\pm} \cup {\cal V} \cup \{\$\}$, 
where $\A$ is the alphabet of $\Gamma$, ${\cal V}$ is the set of variables of $\Gamma$, and $\$$ is a fresh symbol. 
If $\gamma = \exists \bar z \bigwedge_{1 \leq i \leq m} (x_i \xrightarrow{L_i} y_i)$ is a C2RPQ in $\Gamma$ and $\lambda$ is the expansion of 
$\gamma$ obtained by expanding $x_i \xrightarrow{L_i} y_i$ into an oriented 
path $\pi_i$ from $x_i$ to $y_i$ with label $w_i\in L_i$, then 
we encode $\lambda$ as the word 
\begin{align*}
	w_{\lambda} = \$ x_1 w_1 y_1 \$ x_2 w_2 y_2 \$ \, \dotsb \, \$ x_m w_m y_m \$ ~ \, \in \, ~ \A_1^* 
\end{align*}
Note how the subword $x_iw_iy_i$ encodes the oriented path $\pi_i$. 
Every position $j\in \{1,\dots,|w_\lambda|\}$ with $w_\lambda[j]\neq \$$ represents a variable in $\lambda$: either $x_i$ or $y_i$ if $w_\lambda[j]=x_i$ or $w_\lambda[j]=y_i$, respectively; 
or the $(\ell+1)$-th variable in the oriented path $\pi_i$ if $w_\lambda[j]$ is the $\ell$-th symbol in the subword $w_i$. 
Hence different positions in $w_\lambda$ could represent the same variable in $\lambda$, \eg, in the encoding $\$xabcy\$$, 
the 5\emph{th} position containing a `$c$' and the 6\emph{th} position containing a `$y$', represent the same variable, namely, the last vertex $y$ of the oriented path.
It is then easy to build, in polynomial time, an NFA $\+A_1$ over $\A_1$ recognizing the language of all such encodings of expansions of 
$\Gamma$. Our automaton $\+D_\Gamma$ is the product of $\+A_1$ and the DA $\+C_\Gamma$ defined below.
In particular, $\+D_\Gamma$ is limited iff $\+C_\Gamma$ is limited over words of the form $w_\lambda$, for $\lambda$ an expansion of $\Gamma$.

Fix a disjunct $\gamma$ of $\Gamma$. 
As in \cite{CGLV00}, 
we consider words over the alphabet $\A_2 := \A_1 \times (2^{{\cal V}} \cup \{\#\})$ of the form $(\ell_1,\alpha_1)\cdots(\ell_n,\alpha_n)$, 
such that $w_\lambda=\ell_1\cdots\ell_n$, for some expansion $\lambda$ of $\Gamma$, and the $\alpha_i$'s are \emph{valid $\gamma$-annotations}, \ie, (1) $\alpha_i=\#$ if $\ell_i=\$$, (2) $\alpha_1,\dots,\alpha_n\in 2^{{\cal V}}$
induce a partition of the variable set ${\cal V}_\gamma$ of $\gamma$, and (3) for each free variable $x \in \+V_\gamma$ there is some
$(\ell_i,\alpha_i)$ such that $\ell_i=x$ and $x\in \alpha_i$. 
It is easy to construct an NFA $\+B_1^\gamma$ of exponential size that given $w=(\ell_1,\alpha_1)\cdots(\ell_n,\alpha_n)$ with $w_\lambda=\ell_1\cdots\ell_n$, 
checks if the $\alpha_i$'s are valid $\gamma$-annotations. 
Note that if the latter holds, then the annotations encode a mapping $h_w$ from ${\cal V}_\gamma$ to the variables of $\lambda$ such that $h_w(\bar x)=\bar x$, 
where $\bar x$ are the free variables of $\gamma$.

Now, given $w=(\ell_1,\alpha_1)(\ell_2,\alpha_2)\cdots(\ell_n,\alpha_n)$ with $w_\lambda=\ell_1\cdots\ell_n$ and the $\alpha_i$'s being valid $\gamma$-annotations, 
it is shown in \cite{CGLV00} that one can construct in polynomial time a 2NFA $\+B_2^\gamma$ 
that checks the existence of an expansion $\lambda'$ of $\gamma$ and a homomorphism $h$ from $\lambda'$ to $\lambda$ consistent with $h_w$. 
For each atom $x \xrightarrow{L} y$ of $\gamma$, the automaton $\+B_2^\gamma$ guesses 
an oriented path $\pi$ in $\lambda$ from $h_w(x)$ to $h_w(y)$ with label $w'\in L$, directly over the encoding $w_\lambda$ 
starting at a position $j_x$ and ending at a position $j_y$ in $\{0,\dots,n\}$ (recall that the head moves in $\{0,\dots,n\}$) with $j_x,j_y>0$, $w[j_x]=(\ell,\alpha)$, $w[j_y]=(\ell',\alpha')$, 
$x\in \alpha$ and $y\in \alpha'$. 
Note that we have two types of transitions: (1) transitions that consume $a\in \A^\pm$ and actually guess an atom of $\pi$, and 
(2) transitions to ``jump'' from position $j$ to $j'$ in $\{0,\dots,n\}$ representing \emph{equivalent} variables of $\lambda$. 
The latter means that $j,j'>0$ and either $w_\lambda[j]$ and $w_\lambda[j']$ represents exactly the \emph{same} variable of $\lambda$, or 
$w_\lambda[j]$ and $w_\lambda[j']$ represent variables $z,z'$ of $\lambda$ such that $z=^{*}_{\lambda} z'$, where $=^{*}_{\lambda}$ 
is the reflexive-transitive closure of the relation induced by the equality atoms in $\lambda$. 

Let $\+D_2^\gamma$ be the 2DA obtained from the 2NFA $\+B_2^\gamma$ by setting to $0$ and $1$ the cost of transitions of type (2) and (1), respectively.  
Hence, for a word $w$ such that the projection of $w$ to $\A_1$ is $w_\lambda$, and the one to $(2^{{\cal V}} \cup \{\#\})$ is a valid $\gamma$-annotation, we have that $\cost_{\+D_2^\gamma}(w)$
is precisely the minimum size of an expansion $\lambda'$ that can be mapped to $\lambda$ via a homomorphism compatible with $h_w$. 
By Proposition \ref{prop:2da-da}, we can construct in exponential time on $\+D_2^\gamma$
a DA $\+C_2^\gamma$ accepting the same language as $\+D_2^\gamma$ and having an exponential number of states, so that for every word $w'$, we have 
$\cost_{\+C_2^\gamma}(w') \leq \cost_{\+D_2^\gamma}(w') \leq f(\cost_{\+C_2^\gamma}(w'))$ for some polynomial function $f$. 
Let $\exists\+C^\gamma$ be the result of taking the product of $\+B_1^\gamma$ and $\+C_2^\gamma$ and then projecting over the alphabet $\A_1$. 
For every expansion $\lambda$ of $\Gamma$, if $\lambda'$ is a minimal size expansion of $\gamma$ such that $\lambda'\to\lambda$, then 
we obtain that $\cost_{\exists\+C^\gamma}(w_\lambda) \leq \|\lambda'\| \leq f(\cost_{\exists\+C^\gamma}(w_\lambda))$. 
We define our desired $\+C_\Gamma$ to be the union of $\exists\+C^\gamma$ over all $\gamma$ in $\Gamma$. 
We have that for every expansion $\lambda$, if $\lambda_{min}$ is a minimal size expansion of $\Gamma$ such that $\lambda_{min}\to\lambda$, 
then $\cost_{\+C_\Gamma}(w_\lambda) \leq \|\lambda_{min}\| \leq f(\cost_{\+C_\Gamma}(w_\lambda))$. 
By Proposition \ref{prop:basic}, item (2), $\Gamma$ is bounded iff $\|\lambda_{min}\|$ is bounded over all $\lambda$. 
The latter condition holds iff $\+C_\Gamma$ is limited over words $w_\lambda$, for all expansion $\lambda$. 
By definition, the latter is equivalent to $\+D_\Gamma$ being limited. 
Summing up, we obtain that $\Gamma$ is bounded iff $\+D_\Gamma$ is limited, as required. 
Note that the whole construction can be done in exponential time. 
\end{proof}

 As a corollary to Proposition \ref{prop:bounded-2da} and Theorem \ref{thm:da-pspace} we obtain the desired upper bound for 
 part (1) of Theorem \ref{theo:main}. 
 
 \begin{corollary} 
 {\sc Boundedness} for UC2RPQs is in \expspace. 
 \end{corollary} 

\smallskip
\noindent
{\bf {\em Size of equivalent UCQs.}} Here we prove part (2) of Theorem \ref{theo:main}. 
Since $\Gamma$ is bounded we have 
from Proposition \ref{prop:bounded-2da} that $\+D_\Gamma$ is limited. 
Then, from Theorem 
\ref{thm:da-pspace} we obtain that the maximum cost that it takes $\+D_\Gamma$ over a word is $N$, where $N$ is exponential on the number of states of $\+D_\Gamma$, 
and thus double-exponential on $\|\Gamma\|$ by construction. Therefore, for every expansion $\lambda$ of $\Gamma$, if $\lambda_{min}$ is a minimal size expansion $\Gamma$ 
such that $\lambda_{min}\to \lambda$, then $\|\lambda_{min}\| \leq f(N)$,  where $f$ is the polynomial function of the proof of Proposition \ref{prop:bounded-2da}. 
In particular, all minimal expansions of $\Gamma$ 
are of size $\leq$ $f(N)$. By Lemma \ref{lemma:equiv-minimal}, the UC2RPQ $\Gamma$ is equivalent to the union of all its minimal expansions. 
The number of such minimal expansions is thus at most exponential on $f(N)$, and hence triple-exponential on $\|\Gamma\|$.  

\subsection{Lower bounds} 
\label{subsec:expspace-hard}

We reduce from the $2^n$-tiling problem, that is, a tiling problem restricted to $2^n$ many columns, which is  \expspace-complete (see, \eg, \cite{CGLV00}). We show that for every $2^n$-tiling problem $T$ there is a CRPQ $\gamma$, computable in polynomial time from $T$, whose number of minimal expansions is essentially the number of solutions to $T$ in the following sense.

\newcommand{\lemTilingEncoding}{For every $2^n$-tiling problem $T$ with $m$ solutions there is a Boolean CRPQ $\gamma$, computable in polynomial time from $T$, such that the number of minimal expansions of $\gamma$ is $O((g(|T|) + m)^{n+1})$ and $\Omega(m)$, for some doubly exponential function $g$.  
Further, $\gamma$ consists of a Boolean CRPQ of the form $\exists x,y \, \bigwedge_{0 \leq i \leq n} (x \xrightarrow{L_i} y)$, where each $L_i$ is given as a regular expression.}
 
\begin{lemma}\label{lem:tiling-encoding}
\lemTilingEncoding
\end{lemma}

As a corollary, this yields an \expspace lower bound for the boundedness problem (part (1) of Theorem~\ref{theo:main}), as well as a triple-exponential lower bound for the size of the UCQ equivalent to any bounded CRPQ (part (3) of Theorem~\ref{theo:main}), since one can produce $2^n$-tiling problems having triply-exponentially many solutions.

\section{Better-behaved Classes of UC2RPQs}
	\label{sec:nice-c2rpq}

Here we present two restrictions of UC2RPQs that exhibit a better behavior 
 in terms of  the complexity of {\sc Boundedness} than the general case, namely, {\em acyclic UC2RPQs of bounded thickness} and {\em strongly connected UCRPQs}. 
 The improved bounds are \pspace and $\Pi_2^P$, respectively, which turn out to be optimal.

\smallskip
\noindent
\textbf{\em Acyclic UC2RPQs of Bounded Thickness.}
For any two distinct variables $x,y$ of a C2RPQ $\gamma$, we denote by $\text{Atoms}_\gamma(x,y)$ the set 
of atoms in $\gamma$ of the form $x \xrightarrow{L} y$ or $y \xrightarrow{L} x$. 
The \defstyle{thickness} of a C2RPQ $\gamma$ is the maximum cardinality of a set of the form $\text{Atoms}_\gamma(x,y)$, for $x,y$ variables of $\gamma$ with $x \neq y$. 
The thickness of a UC2RPQ $\Gamma$ is the maximum thickness over all the C2RPQs in $\Gamma$. 
The \emph{underlying undirected graph} of $\gamma$ has as vertex set the set of variables of $\gamma$ and contains an edge 
$\{x,y\}$ iff $x\neq y$ and $\text{Atoms}_\gamma(x,y) \neq \emptyset$.
A C2RPQ $\gamma$ is \emph{acyclic} if its underlying undirected graph is an acyclic graph (\ie, a forest). 
A UC2RPQ $\Gamma$ is acyclic if each C2RPQ in $\Gamma$ is. 

We show next that {\sc Boundedness} for acyclic UC2RPQs of bounded thickness is \pspace-complete. These classes of UC2RPQs have 
been previously studied in the literature~\cite{BRV14,BRV16}. 
In particular, it follows from~\cite[Theorem 4.2]{BRV16} that the containment problem for the acyclic UC2RPQs of bounded thickness is \pspace-complete, and hence 
Theorem~\ref{theo:strongly-acyclic} below 
shows that {\sc Boundedness} is not more costly than containment for these classes.  

\begin{theorem}
\label{theo:strongly-acyclic}
Fix $k\geq 1$. The problem {\sc Boundedness} is \pspace-complete for acyclic UC2RPQs of thickness at most $k$. 
\end{theorem}

\begin{proof}[Proof (sketch)]
The lower bound follows directly from \pspace-hardness of {\sc Boundedness} for RPQs (see Corollary~\ref{coro:comp-rpq}). 
For the \pspace upper bound, we follow a similar strategy as in the case of arbitrary UC2RPQs (Section~\ref{sec:upper-bound}), \ie, 
we reduce boundedness of $\Gamma$ to DA limitedness. 
The main difference is that, since $\Gamma$ is acyclic, we can exploit the power of alternation and construct an A2DA$^\epsilon$ $\+B$ (instead of a 2DA, as in the proof 
of Proposition \ref{prop:bounded-2da}), 
such that $\Gamma$ is bounded iff $\+B$ is limited. The constant upper bound on the thickness of $\Gamma$ implies that $\+B$ is actually of polynomial size. 
The result follows then as limitedness of an A2DA$^\epsilon$ can be decided in \pspace in virtue of Theorem~\ref{thm:a2da-limitedness-pspace}. 
\end{proof} 

Both conditions in Theorem~\ref{theo:strongly-acyclic}, \ie, acyclicity and bounded thickness, are necessary. 
Indeed, it follows from Lemma~\ref{lem:tiling-encoding} that {\sc Boundedness} is \expspace-hard even for:
\begin{itemize} 
\item Boolean acyclic CRPQs. 
\item Boolean CRPQs of thickness one, whose underlying undirected graph is of {\em treewidth} two. Recall that the treewidth is a measure of how 
much a graph resembles a tree (cf., \cite{diestel}) ---acyclic graphs are precisely the graphs of treewidth one.  
\end{itemize} 
Indeed, the CRPQs of the form $\exists x,y \bigwedge_{i} (x \xrightarrow{L_i} y)$ used in Lemma ~\ref{lem:tiling-encoding} are Boolean and acyclic (but have unbounded thickness).  Replacing each $(x \xrightarrow{L_i} y)$ with $(x \xrightarrow{\epsilon} z_i) \land (z_i \xrightarrow{L_i} y)$, yields an equivalent CRPQ of thickness one whose 
underlying undirected graph has treewidth two. 

\smallskip
\noindent
\textbf{\em Strongly Connected UCRPQs.}
We conclude this section with an even better behaved class of CRPQs in terms of {\sc Boundedness}. Unlike the previous case, the definition of this 
class depends on the underlying {\em directed} graph of a CRPQ $\gamma$. This contains a directed edge from variable $x$ to $y$ iff 
there is an atom in $\gamma$ of the form $x \xrightarrow{L} y$. 
A CRPQ $\gamma$ is \defstyle{strongly connected} if its underlying directed graph is strongly connected, \ie, every pair of variables is connected by some directed path. 
A UCRPQ $\Gamma$ is strongly connected if every CRPQ in $\Gamma$ is. We can then establish the following. 

\begin{theorem}
\label{theo:strongly-connected}
 {\sc Boundedness} is $\Pi_2^P$-complete for strongly connected UCRPQs.
\end{theorem}

\section{Discussion and Future Work}
	\label{sec:discussion}




The main conclusion of our work is that techniques previously used in the study 
of containment of UC2RPQs can be naturally leveraged to pinpoint the complexity of {\sc Boundedness} 
by using DA instead of NFA. This, however, requires extending results on limitedness to alternating 
and two-way DA. For all the classes of UC2RPQs studied in the paper we show in fact that 
the complexity of {\sc Boundedness} coincides with that of the containment problem. We leave open what is the exact 
size of UCQ rewritings for the classes of acyclic UC2RPQs of bounded thickness and the strongly connected UCRPQs that are bounded.  

The most natural next step is to study {\sc Boundedness} for the class of {\em regular queries} (RQs), which are 
the closure of UC2RPQs under binary transitive closure. 
RQs are one of the most powerful recursive languages for which containment is decidable in elementary time. 
In fact, containment of 
RQs has been proved to be {\sc 2EXPSPACE}-complete 
by applying sophisticated techniques based on NFA \cite{RRV17}. 
We will study if it is possible to settle the complexity of 
{\sc Boundedness} for RQs with the help of DA techniques.

\bibliography{long,biblio}

\newpage
\appendix


\section{Appendix to Section \ref{sec:char}}

\begin{proof}[Proof of Lemma \ref{lemma:equiv-minimal}]
The result is a straightforward consequence of the following 
known result.

\begin{lemma} \label{lemma:cont-crpq} {\em \cite{CGLV00}} Let $\Gamma,\Gamma'$ be UC2RPQs. It is the case that 
$\Gamma \subseteq \Gamma'$ iff for each expansion $\lambda$ of $\Gamma$ there exists an expansion $\lambda'$ of 
$\Gamma'$ such that 
$\lambda \subseteq \lambda'$, or, equivalently, 
$\lambda' \to \lambda$.  
\end{lemma}

\end{proof} 

\begin{proof}[Proof of Lemma \ref{prop:basic}]
For (1)$\Rightarrow$(2), suppose that $\Gamma$ is equivalent to a UCQ $\Phi = \bigvee_{1 \leq i \leq n} \phi_i$. 
By Lemma \ref{lemma:cont-crpq}, for every $1\leq i \leq n$, there is an expansion $\lambda_i$ of $\Gamma$ with $\lambda_i\to \phi_i$. 
We claim that (2) holds for $k=\max\{\|\lambda_i\|: 1\leq i \leq n\}$. 
Let $\lambda$ be an expansion of $\Gamma$. By Lemma \ref{lemma:cont-crpq}, we have $\phi_i\to \lambda$, 
for some $1\leq i \leq n$, and then $\lambda_i\to \lambda$. Since $\|\lambda_i\|\leq k$, we are done. 
For the implication (2)$\Rightarrow$(3), note that if (2) holds for some $k \geq 1$, then the size of any minimal expansion of $\Gamma$ is at most $k$. 
Finally, (3)$\Rightarrow$(1) follows directly from Lemma \ref{lemma:equiv-minimal}. 
\end{proof}

\section{Appendix to Section~\ref{sec:sec-rpq-bound}} 

\begin{proof}[Proof of Proposition \ref{prop:basic-rpq} ]
Note that for an RPQ or an existentially quantified RPQ $\gamma$ whose input regular language is $L$, 
there is a bijection from $L$ to the expansions of $\gamma$. 
For an RPQ $L$ every expansion is minimal. 
In the case of a CRPQ $\gamma(x)=\exists y (x \xrightarrow{L} y)$ \resp{$\gamma(y)=\exists x (x \xrightarrow{L} y)$}, where $x\neq y$, there is a bijection from $L_{\text{pf}}$ \resp{$L_{\text{sf}}$} to the minimal expansions of $\gamma$. 
Finally, for a Boolean CRPQ $\gamma=\exists x,y ( x \xrightarrow{L} y)$, with $x\neq y$,
we have a bijection from $L_{\text{ff}}$ to the minimal expansions. 
Then, the proposition follows directly from Proposition~\ref{prop:basic}, item (3). 
\end{proof}

 \begin{proof}[Proof of Theorem \ref{theo:finite-lang}]
We focus on the lower bounds. 
We reduce from the following well-known \pspace-complete problem: given a non-deterministic Turing machine $M$ and a natural number $n$ (given in unary) check whether $M$ accepts the empty tape using $n$ space. 
 As usual, we encode configurations of $M$ as words of length $n$ over the alphabet $\mathbb{P}:=\Sigma\cup (\Sigma\times S)$, where $\Sigma$ and $S$ are the tape alphabet and state set of $M$ respectively. 
A run of $M$ is then encoded by a word of the form $\# c_1\cent c_2 \cent \dotsb \cent c_\ell \#$, where each $c_i$ is an encoding of a configuration, $\cent$ is used as a separator of configurations, and $\#$ to delimit the beginning and end of the run. 
We can assume without loss of generality that either $M$ accepts the empty tape using $n$ space and all non-deterministic branches in the computation of $M$ accept before $|\mathbb{P}|^n$ steps;
or any non-deterministic branch of $M$ does not halt at all. 

Given $M$ and $n$ as above, we can define an NFA of polynomial size that accepts the language $R:=c_{\text{init}}\cdot (\cent\cdot C)^*$ over $\mathbb{P}$, 
where $c_{\text{init}}$ encodes the (unique) initial configuration of $M$ on the empty tape and $C$ accepts all the words of length $n$ over $\mathbb{P}$ that encode a configuration of $M$. 
We let $T$ be the finite language that contains all the words of the form $c\cent c'$, where $c$ and $c'$ encode configurations of $M$ and $c'$ \emph{cannot} be reached from $c$ in one step. 
Note that $T$ can be accepted by a polynomial-sized NFA. 
We claim that $M$ accepts the empty tape iff $L_{\text{pf}}$ is finite for $L:=\# R \# \cup \#(R\cent+\varepsilon) T$. 

Assume first that $M$ accepts the empty tape. Then every $w\in L_{\text{pf}}$ satisfies $|w|\leq 2+(n+1)(|\mathbb{P}|^n+1)$. 
Indeed, by contradiction, suppose that this is not the case for some $w\in L_{\text{pf}}$. 
If $w\in \# R \#$, then $w$ cannot encode a run of $M$ starting from the empty tape (as every such a run takes less than $|\mathbb{P}|^n$ steps). 
Then there exists $v\in \#(R\cent+\varepsilon) T$ with $|v|\leq 1+(n+1)(|\mathbb{P}|^n+1)$ that is a prefix of $w$. In particular, $|v|<|w|$ and then $w\not\in L_{\text{pf}}$; 
a contradiction. Similarly, if $w\in \#(R\cent+\varepsilon) T$, then $w\#$ cannot encode a run of $M$ starting from the empty tape, 
and hence $w\not\in L_{\text{pf}}$. 
Now suppose that $M$ does not accept the empty tape. Then there are infinitely many words $w\in \# R \#$ 
encoding a run of $M$ starting from the empty tape, which in particular belong to $L_{\text{pf}}$. 
Hence, $L_{\text{pf}}$ is infinite. 

Note that the same construction applies for the factor case, \ie, $M$ accepts the empty tape iff $L_{\text{ff}}$ is finite for $L=\# R \# \cup \#(R\cent+\varepsilon) T$ 
(a simpler construction that still applies is $L=\# R \# \cup T$). 
Finally, we can reduce the prefix to the suffix case. Given an NFA accepting $L$, we construct an NFA accepting $L^R=\{w^R: w\in L\}$, where $w^R$ is the \emph{reverse} of $w$. 
Hence, $L_{\text{pf}}$ is finite iff $L^{R}_{\text{sf}}$ is finite. 
 \end{proof}

\section{Appendix to Section~\ref{sec:distance}}

\begin{proof}[Proof of Proposition~\ref{prop:2da-da}]
We adapt the standard ``\textsl{crossing sequence}'' construction for translating two-way NFA to one-way NFA \cite{shepherdson1959reduction}. For simplicity, and without any loss of generality, we assume that the input 2DA $\+A$ is so that all accepting runs end with the head at the leftmost position, and that there is only one final state.

Given a 2DA $\+A = (\A, Q, \emptyset, q_0, \set{q_f}, \delta)$ consider the DA $\+B = (\A, Q', \emptyset, (q_0, q_f), F', \delta' )$ where  $Q'$ is the set of subsets of $Q \times Q$, and 
$F'$ is the set of subsets of $\set{(q,q) : q \in Q}$. The idea is that whenever a state contains a pair $(q,p)$, it verifies that there is a loop at the current position in the run, starting in $q$ and ending in $p$ and visiting only positions which are to the right. Formally, there is a transition $(S, a, c, S') \in \delta'$ if for every pair $(q,p) \in S$ with $q\neq p$ there are  $(r_1, r'_1), \dotsc, (r_n, r'_n) \in S'$ so that
\begin{itemize}
	\item there is a transition $(q,a,c',r_1) \in \delta$ and $(r'_n,a^{-1},c',p) \in \delta$;
	\item for every $r'_i$ with $i<n$ there are transitions $(r'_i,a^{-1},c_i,s_i), (s_i,a,c_{i+1},r_{i+1}) \in \delta$ for some $s_i \in Q$ and $c_i,c_{i+1} \in \set{0,1}$;
\end{itemize}
and the maximum of the costs of the considered transitions is $c$. The following figure exemplifies the relation between $(p,q)$ and the $(r_i,r'_i)$'s as seen in a run.
\begin{center}
\includegraphics[width=.17\textwidth]{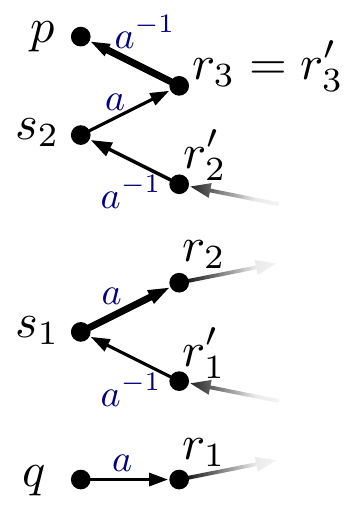}	
\end{center}
In this example, the loop $(q,p)$ is witnessed, at the next position after reading an $a$, as the existence of two loops $(r_1,r'_1)$ and $(r_2,r'_2)$ (and the trivial loop $(r_3,r'_3)$).

From this construction it follows that $\+B$ accepts the same language as $\+A$, and that $\cost_{\+B}(w) \leq \cost_{\+A}(w)$ for every $w \in \A^*$. Further, note that without any loss of generality we can consider only states having no to distinct pairs with the same state, and thus states with at most $|Q|$ pairs. From this, it follows that $\cost_{\+A}(w) \leq \cost_{\+B}(w) \cdot 2 \cdot |Q|^2$. In the picture above, we are simulating 2 costly transitions (depicted with thick strokes) with only one costly transitions for the pair $(q,p)$, and in general we could see $2|Q|$ transitions for each pair $(q,p)$ in the state, hence $2|Q| \cdot |Q|$ costly transitions could be simulated at once. In fact, a finer analysis can show that $\cost_{\+A}(w) \leq \cost_{\+B}(w) \cdot 2 \cdot |Q|$.
\end{proof}

\paragraph*{Proof of Theorem~\ref{thm:a2da-limitedness-pspace}}

Here we give the detailed proof of Theorem~\ref{thm:a2da-limitedness-pspace}, namely, that the limitedness problem for A2DA is in \pspace.


\paragraph*{(1) From A2DA$^\epsilon$ to A2DA}
We show a language-preserving polynomial reduction from A2DA$^\epsilon$ to A2DA, obtained by replacing $\epsilon$-transitions with a sequence of two transitions reading $a$ and $a^{-1}$ for some $a \in \A$. Formally, given a A2DA$^\epsilon$ $\+A = (\A, Q_\forall, Q_\exists, q_0, F, \delta)$, we produce a A2DA $\+B = (\A, Q'_\forall, Q'_\exists, q_0, F', \delta')$ where $Q'_\forall$ contains $Q_\forall$ \resp{$Q'_\exists$ contains $Q_\exists$}, plus fresh states $r_{q,p}$ and $r'_{q,p}$ for each $q \in Q_\forall$ \resp{$q \in Q_\exists$} and $p \in Q_\forall \cup Q_\exists$.  The transition relation $\delta'$ is obtained from $\delta$ by replacing each transition $(q,\epsilon,e,c,p)$ with all (polynomially many) transitions $(q,a,e',c,r_{q,p})$, $(q,a^{-1},e,0,r_{q,p})$, $(q,a^{-1},e',c,r'_{q,p})$, $(r'_{q,p},a,e,0,p)$ for every $a \in \A$ and $e' \in \set{0,1}$, where $r_{q,p}$ and $r'_{q,p}$ are fresh states of the same type as $q$ (\ie, $r_{q,p} \in Q'_\forall$ iff $q \in Q_\forall$ and likewise for $r'_{q,p}$). In other words, either we simulate the $\epsilon$-transition by reading $a \cdot a^{-1}$ for some $a \in \A$ through $r_{q,p}$, or we simulate it by reading $a^{-1}\cdot a$ through $r'_{q,p}$. Note that if we would now define $F'=F$ we may not accept the empty word because we need words of length at least 1 to simulate $\epsilon$-transitions. In order to fix this, we define $F' = F \cup \set{q_0}$ if some state of $F$ can be reached from $q_0$ through a sequence of $\epsilon$-transitions with $\yend$ flag, or we define $F' = F$ otherwise.

The above reduction, although it does preserve the language, it does not preserve the cost of words: while the cost of the empty word can only be $0$ for any A2DA, for an A2DA$^\epsilon$ automaton it can be any arbitrary $n \in \N$. However, since it is a faithful simulation on non-empty words, for all $w \in \A^+$ we have $\cost_{\+A}(w) = \cost_{\+B}(w)$, and thus the reduction preserves the limitedness property.

\paragraph*{(2) From A2DA to 2DA}

Given a A2DA $\+A = (\A, Q_\forall, Q_\exists, q_0, F, \delta)$, we build a 2DA $\+B$ over a larger alphabet $\B$, where we trade alternation for extra alphabet letters. The alphabet $\B$ consists of triples $(f^\rightarrow,a,f^\leftarrow)$ where $a \in \A$ and $f^\rightarrow,f^\leftarrow : Q_\forall \to \delta$. The idea is that $f^\rightarrow,f^\leftarrow$ are ``choice functions'' for the alternation: whenever we are to the left \resp{right}  of a position of the word labelled $(f^\rightarrow,a,f^\leftarrow)$ and we are in state $q \in Q_\forall$, instead of exploring all transitions departing from $q$ and taking the maximum cost over all such runs (this is what alternation does), we chose to just take transition $f^\rightarrow(q)$ \resp{$f^\leftarrow(q)$}. Note that $\B$ is exponential in the number of states but not in the size of $\A$.

One can then build a 2DA $\+B$ having the same set of states as $\+A$ but with a transition function which is essentially deterministic on the states of $Q_\forall$, as it follows the choice function given by the alphabet letters. In the end we obtain that
\newcommand{\extrategy}[1]{\widetilde{#1}}%
\begin{itemize}
	\item for every $w \in \B^*$, $\cost_{\+B}(w) \leq \cost_{\+A}(w_\A)$; and
	\item for every $w\in \A^*$ there is $\extrategy w\in \B^*$ so that $\extrategy w_\A=w$ and $\cost_{\+A}(w) = \cost_{\+B}(\extrategy w)$,\footnote{This can be alternatively seen as the existence of a positional strategy for the universal player for obtaining the cost $\cost_{\+A}(w)$ when the automata is seen as a two-player game.}
\end{itemize}
where $w_\A$ is the projection of $w$ onto $\A$. This shows that the limitedness problem is preserved.

For simplicity we assume that any state of $Q_\forall$ determines whether the head of the automaton moves rightwards or leftwards ---it is easy to see that this is without loss of generality. Formally, we have that $Q_\forall$ is partitioned into two sets $Q_\forall = Q_\forall^\rightarrow \dcup Q_\forall^\leftarrow$ so that there is no transition $(r,a,c,p) \in \delta^\rightarrow$ with $r \in Q_\forall^\leftarrow$, and no transition $(r,a,c,p) \in \delta^\leftarrow$ with $r \in Q_\forall^\rightarrow$, where $\delta^\rightarrow$ \resp{$\delta^\leftarrow$} is the set of all transitions from $\delta$ reading a letter from $\A$ \resp{from $\A^{-1}$}.

More concretely, consider the alphabet $\B = \set{f^\rightarrow : Q^\rightarrow_\forall \to \delta} \times \A \times \set{f^\leftarrow : Q^\leftarrow_\forall \to \delta}$, and notice that
\begin{align}
	|\B| &\leq |\A| \cdot |\delta|^{2\cdot|Q_\forall|}.\label{eq:alphabet-blowup-redux}\tag{$\star$}
\end{align}

We now define a 2DA $\+B$ so that $\+A$ (over $\A$) is limited if, and only if, $\+B$ (over $\B$) is limited. $\+B$ has the same set of states as $\+A$ but its transition function is essentially deterministic on the states of $Q_\forall$.

Concretely, let $\+B$ be the 2DA defined as $(\B, \emptyset,Q_\forall \cup Q_\exists, q_0, F, \delta')$, where $\delta'$ is the union of
\begin{itemize}
	\item $\set{(r,(f^\rightarrow,a,f^\leftarrow),e,c,p) : r \in Q_\exists \land (r,a,e,c,p) \in \delta \land (f^\rightarrow,a,f^\leftarrow) \in \B }$, that is, all transitions from $\delta$ starting from an existential state moving rightwards;
	\item $\set{(r,(f^\rightarrow,a,f^\leftarrow),e,c,p) : r \in Q^\rightarrow_\forall \land f^\rightarrow(r) = (r,a,e,c,p) \land (f^\rightarrow,a,f^\leftarrow) \in \B}$, that is, for any universal state, the transition defined by $f^\rightarrow$ moving rightwards if it can be applied;
	\item $\set{(r,(f^\rightarrow,a,f^\leftarrow),e,c,p) : r \in Q^\rightarrow_\forall \land f^\rightarrow(r) = (r',a',e',c',p') \land (r,a,e) \neq (r',a',e') \land (r,a,e,c,p) \in \delta \land (f^\rightarrow,a,f^\leftarrow) \in \B}$, that is, if for a universal state $r$, $f^\rightarrow(r)$ gives an inconsistent transition, disregard $f^\rightarrow$ and take any (consistent) transition from $\delta$;
\end{itemize}
and similar sets for the case of the transitions moving leftwards:
\begin{itemize}
	\item $\set{(r,(f^\rightarrow,a,f^\leftarrow)^{-1},e,c,p) : r \in Q_\exists \land (r,a^{-1},e,c,p) \in \delta \land (a,f^\rightarrow,f^\leftarrow)^{-1} \in \B^{-1}}$;
	\item $\set{(r,(f^\rightarrow,a,f^\leftarrow)^{-1},e,c,p) : r \in Q^\leftarrow_\forall \land f^\leftarrow(r) = (r,a^{-1},e,c,p) \land (f^\rightarrow,a,f^\leftarrow)^{-1} \in \B^{-1}}$;
	\item $\set{(r,(f^\rightarrow,a,f^\leftarrow)^{-1},e,c,p) : r \in Q^\leftarrow_\forall \land f^\leftarrow(r) = (r',a',e',c',p') \land (r,a^{-1},e) \neq (r',a',e') \land (r,a^{-1},e,c,p) \in \delta \land (f^\rightarrow,a,f^\leftarrow)^{-1} \in \B^{-1}}$.
\end{itemize}
For any word $w \in \B^*$ we denote by $w_\A \in \A^*$ its projection onto $\A$.

\begin{lemma}\label{lem:a2da-2da:easylemma}
For every $w \in \B^*$, $\cost_{\+B}(w) \leq \cost_{\+A}(w_\A)$.
\end{lemma}
\begin{proof}
Let $w$ be an arbitrary word over $\B$, $w=(f^\rightarrow_{1},a_1,f^\leftarrow_{1}) \dotsb (f^\rightarrow_{n},a_n,f^\leftarrow_{n}) \in \B^*$. If there is no accepting run of $\+A$ on $w_\A$, then there is no accepting run of $\+B$ on $w$. \sidediego{explain}
Otherwise, suppose there is an accepting run $t$ of $\+A$ on $w_\A = a_1 \dotsb a_n$ of cost $N$. We show that the functions $f^\rightarrow_{i}$'s and $f^\leftarrow_{i}$'s allow us to select a branch of $t$ whose labelling $c_1 \dotsb c_m \in \delta^*$ can be extended to an accepting run $c'_1 \dotsb c'_m \in \delta'^*$ of $\+B$ on $w$ of cost $\leq N$. In this way, it follows that $\cost_{\+B}(w) \leq N \leq \cost_{\+A}(w_\A)$.

Let us see how to obtain such a branch. Let us fix any homomorphism $h: t \to G_{\+A,w_\A}$, given by the fact that $t$ is an accepting run.
Consider the following traversal of $t$, starting at the root. Whenever we are at a node $x$ with $h(x) \in Q_\exists \times \set{0, \dotsc, |w|}$, we go to the only child (unless it is the leaf, in which case the traversal ends); and whenever we are at a node $x$ with $h(x) = (r,i) \in Q_\forall^\rightarrow \times \set{0, \dotsc, |w|}$ \resp{$(r,i) \in Q_\forall^\leftarrow \times \set{0, \dotsc, |w|}$}, we go to the child obtained when taking the edge labelled $f^\rightarrow_{i}(r)$ \resp{$f^\leftarrow_{i}(r)$} if there is one, or to any child otherwise.
Consider now the labeling $c_1 \dotsb c_m \in \delta^*$ corresponding to the branch of $t$ just described, and let $b_1 \dotsb b_m \in (\A \cup \A^{-1})^*$ be the letters read by the transitions. For each $i \in \set{1, \dotsc, m}$ with $b_i \in \A$ \resp{$b_i = a^{-1}$ with $a \in \A$}, we define $c'_i \in \delta'$ as the result of replacing $b_i$ with $(f^\rightarrow_{\ell},b_i,f^\leftarrow_{\ell})$ \resp{replacing $a^{-1}$ with $(f^\rightarrow_{\ell},a,f^\leftarrow_{\ell})^{-1}$} in $c_i$, for $\ell = |b_1 \dotsb b_i|$ . It follows that $c'_1 \dotsb c'_m$ is an accepting run of $\+B$ on $w$ with cost the same cost as $c_1 \dotsb c_m$, which must be at most $N$.
\end{proof}

\begin{lemma}\label{lem:a2da-2da:strategy}
For every $w\in \A^*$ there is $\extrategy w\in \B^*$ so that $\extrategy w_\A=w$ and $\cost_{\+A}(w) = \cost_{\+B}(\extrategy w)$.
\end{lemma}
\begin{proof}
Given a word $w \in \A^*$, we define a function $f_{\+A,w} : \set{0, \dotsc, |w|} \to Q_\forall \to \delta$ that maximizes the cost for $w$. 
That is, $f_{\+A,w}(i)(q)$ is the transition that should be followed whenever we are in state $q$ at position $i$ in order to obtain $\cost_{\+A}(w,q,i)$ ---\ie, to maximize the cost. 
(This function can be regarded as a positional strategy for the universal player obtaining the maximum cost in the two-player game associated with the A2DA automaton.) 
Formally, for any $q \in Q_\forall$, we define $f_{\+A,w}(i)(q)$ as any transition $t \in \delta$ so that $t \in \arg \max g_{i,q}$ for $	g_{i,q} : \delta \to \N \cup \set{\infty,-1}$ defined as
\begin{align*}
	g_{i,q}(r,a,c,p)&=
\begin{cases}
	c + \cost_{\+A}(w,p,i+1) & \text{ if } i<|w|, a = w[i+1] \in \A, r=q;\\
    c + \cost_{\+A}(w,p,i-1) & \text{ if } i>0, a = w[i]^{-1} \in \A^{-1}, r=q;\\
	-1                       & \text{ otherwise.}
\end{cases}
\end{align*}

For any  word $u \in \A^*$, let $\extrategy u \in \B^*$ be so that $|\extrategy u|=|u|$ and  for each $1 \leq i \leq |u|$ we define $\extrategy u[i] = (f^\rightarrow, u[i], f^\leftarrow)$, \sidediego{check notation $w[i]$} where $f^\rightarrow$ \resp{$f^\leftarrow$} is the restriction of $f_{\+A,u}(i-1)$ on the subdomain $Q^\rightarrow_\forall$ \resp{of $f_{\+A,u}(i)$ on $Q^\leftarrow_\forall$}. By definition of $f_{\+A,u}$ and $\+B$, it follows that $\cost_{\+A}(w) = \cost_{\+B}(\extrategy w)$.	
\end{proof}

\begin{lemma}\label{lem:}
$\+A$ is limited if, and only if, $\+B$ is limited.
\end{lemma}
\begin{proof}
If $\+B$ is limited, there is some $N \in \N$ so that $\cost_{\+B}(w) \leq N$ for every $w \in \B^*$. By Lemma~\ref{lem:a2da-2da:strategy}, for every $w \in \A^*$ there is some $\extrategy w \in \B^*$ so that $\cost_{\+A}(w) = \cost_{\+B}(\extrategy w) \leq N$ and thus $\+A$ is limited by $N$.

If $\+A$ is limited, then there is some $N \in \N$ so that $\cost_{\+A}(w) \leq N$ for every $w \in \A^*$. Then, for every word $u \in \B^*$ we have $\cost_{\+B}(u) \leq \cost_{\+A}(u_\A) \leq N$ by Lemma~\ref{lem:a2da-2da:easylemma}.
\end{proof}

\paragraph*{(3) From 2DA to ADA${}^\epsilon$}
We show a polynomial-time translation from 2DA to ADA extended with $\epsilon$-transitions, which preserves limitedness. In the case of finite automata, there exist language-preserving reductions from 2-way NFA to alternating 1-way NFA with a quadratic blowup \cite{birget1993state,PitermanV03}. However, these translations, when applied blindly to reduce from 2DA to ADA, do not preserve the cost semantics nor limitedness of languages. On the other hand, \cite{BlumensathCKPB14} shows an involved construction that results in a reduction from 2DA to ADA on \textsl{infinite trees} (a more general and challenging setup), which preserves limitedness but it is not polynomial in the number of states. Here we give a self-contained translation from 2DA to ADA which serves our purpose: it preserves limitedness and it is polynomial time computable.

The translation is close to the language-preserving reduction from 2NFA to alternating 1NFA of \cite{PitermanV03}, upgraded to take into account the cost of different alternation branches, somewhat in the same spirit as the \textsl{history summaries} from \cite{BlumensathCKPB14}.
The reduction exploits the structure of the runs of 2-way finite automata, which can be described as ``a tree  of zig-zags'', borrowing the wording of \cite{PitermanV03}. That is, every  run of a two-way finite automaton on $w$ can be seen as a tree of height at most $|w|$ whose nodes are labelled by $Q \cup Q^2$, in such a way that for each letter $a \in \A$, each pair of consecutive transitions reading $a, a^{-1}$ induce a leaf and each pair of consecutive transitions reading $a^{-1}, a$ induce a branching (\cf~Figure~\ref{fig:2way-as-tree}). The idea is then to explore the tree top-down by spawning new threads at every branching and using only a statespace of $Q \cup Q^2$. Let us call a \defstyle{zig-zag tree} to any such tree resulting from an accepting run of a 2DA.
\begin{figure}[t]
	\includegraphics[width=\textwidth]{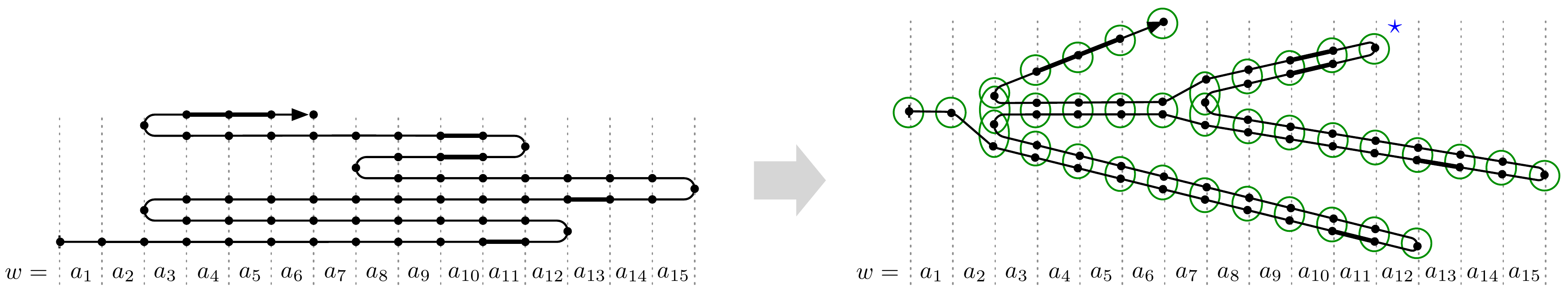}
	\caption{A run of 2DA seen as a tree of zig-zags. Tree nodes are depicted as circles, and there is an (implicit) edge from any circle to its right neighbor(s). Thick strokes represent costly transitions.}\label{fig:2way-as-tree}
\end{figure}

Concretely, given a 2DA $\+A = (\A, Q, \emptyset, q_0, F, \delta)$ we construct a ADA${}^\epsilon$ \[\+B = (\A, Q_\exists\cup\{(q_0,\yend)\}, Q_\forall, (q_0,\yend), F', \delta')\] where $Q_\exists = (Q \cup Q^2) \times \set{0,1} \times \set{\yend,\nend}$ 
and $Q_\forall$, $F'$ and $\delta'$ are of polynomial size. The idea is that during the run, a state $(q,p,c,e) \in Q \times Q \times \set{0,1} \times \set{\yend,\nend}$ ---which we henceforth note as $[q,p]^e_c$--- at position $i$ verifies the presence of a `right loop', that is, a partial run of $\+A$ that starts in state $q$ at $i$ and ends in state $p$ at $i$, visiting only positions $j\geq i$ to the right of $i$. The subscript $c$ states whether the looping run contains at least one costly transition ($c=1$) or no costly transitions ($c=0$). The superscript $e$ is simply a flag with the information of whether the current position is and end position or else, which we need for technical reasons. A state $\alpha \land \beta$ is understood as the alternation of states $\alpha$ and $\beta$. We will build a ADA with $\epsilon$-transitions (\ie, a ADA${}^\epsilon$), and for this reason we allow to have transitions $(q,a,e,2,p)$ as short for $(q,a,e,1,p'), (p',\epsilon,e,1,p)$ for a fresh state $p'$. Formally, $\+B$ is defined as follows:
\begin{itemize}
	\item $Q_\forall = \set{ \alpha \land \beta : \alpha,\beta \in Q_\exists}$;
	\item $F' = \set{[q]^e_0 : q \in F, e \in \set{\yend,\nend}} \cup \set{[q,q]^e_0 : q \in Q, e \in \set{\yend,\nend}}$;
	\item $\delta'$ is defined as the smallest set verifying
	\begin{itemize}
		\item $((q_0,\yend),\epsilon, \yend, 0, [q_0]_c^{\yend})\in \delta'$, for every $c\in \{0,1\}$;
		\item $\set{([q]^{e'}_1,a,e,1,[p]^e_c) : c \in \set{0,1}, a \in \A, e,e' \in \set{\yend,\nend}, (q,a,e,1,p) \in \delta} \subseteq \delta'$ ---\ie, every costly rightward transition of $\delta$ is in $\delta'$;
		\item $\set{([q]^{e'}_c,a,e,0,[p]^e_c) : c \in \set{0,1}, a \in \A, e,e' \in \set{\yend,\nend}, (q,a,e,0,p) \in \delta} \subseteq \delta'$ ---\ie, every non-costly rightward transition of $\delta$ is in $\delta'$; 
		\item for every $q,p \in Q$, $c_1, c_2 \in \set{0,1}$, and $e \in \set{\yend,\nend}$ we have $([q]^e_{\max(c_1,c_2)},\epsilon,e,0,[p]^e_{c_1} \land [q,p]^e_{c_2})\in\delta'$ ---\ie, we can change the state from $q$ to $p$ provided there is a right loop from $q$ to $p$;
		\item for every $q,p,r \in Q$, $c_1,c_2 \in \set{0,1}$, and $e \in \set{\yend,\nend}$, we have $([q,r]^e_{\max(c_1,c_2)},\epsilon,e,0,[q,p]^e_{c_1} \land [p,r]^e_{c_2})\in\delta'$ ---\ie, there is a right loop from $q$ to $r$ if there are from $q$ to $p$ and from $p$ to $r$;
		\item for every $a \in \A$, $q,p \in Q$, $c \in \set{0,1}$, $e_1,e_2 \in \set{\yend,\nend}$ and $(q,a,e_1,c_1,q'), (p',a^{-1},e_2,c_2,p) \in \delta$ we have
		\begin{itemize}
			\item $([q,p]^{e_2}_c,a,e_1,0,[q',p']^{e_1}_c)\in\delta'$ if $\max(c_1,c_2)=0$,
			\item $([q,p]^{e_2}_1,a,e_1,c_1 + c_2,[q',p']^{e_1}_c)\in\delta'$ if $\max(c_1,c_2)=1$;
		\end{itemize}
---\ie, a right loop on a position can be witnessed by a right loop on the next position;
		\item for every  $[\alpha_1]^{e}_{c_1} \land [\alpha_2]^{e}_{c_2} \in Q_\forall$ we have $([\alpha_1]^{e}_{c_1} \land [\alpha_2]^{e}_{c_2}, \epsilon, e, c_2, [\alpha_1]^{e}_{c_1}) \in\delta'$ and $([\alpha_1]^{e}_{c_1} \land [\alpha_2]^{e}_{c_2}, \epsilon, e, c_1, [\alpha_2]^{e}_{c_2}) \in \delta'$ ---\ie, $\alpha_1 \land \alpha_2$ is the alternation of states $\alpha_1$ and $\alpha_2$.
	\end{itemize}
\end{itemize}

\begin{lemma}\label{lem:}
	$\+A$ is limited if, and only if $\+B$ is limited.
\end{lemma}
\begin{proof}
First, note that the translation is language-preserving (\ie, the set of words with accepting runs in $\+A$ and in $\+B$ coincide). Further, the accepting runs of $\+B$ are essentially the accepting runs of $\+A$ seen as zig-zag trees.

For any accepting run of $\+A$ represented as a zig-zag tree, and any given branch thereof (\ie, a path from the root to a leaf), let us define its \defstyle{cost} as the number of costly transitions it contains. For example, in Figure~\ref{fig:2way-as-tree} the branch indicated with $\star$ has cost $2$. We also define the \defstyle{number of heavy branchings} of a branch as the number of subtrees attached to the branch that have at least one costly transition. In Figure~\ref{fig:2way-as-tree} there are 3 subtrees attached to the $\star$-branch (which, in this particular case, they all look like words rather than trees), and all of them have costly transitions; hence the number of heavy branchings is 3. Finally, for any accepting run $\rho$ of $\+A$, let $f(\rho)$ be the maximum, over all its branches, of its cost plus its number of heavy branchings. Notice that $f(\rho) \leq \cost(\rho)$. 
Observe also that for a word $w$, every accepting run $t$ of $\+B$ determines a zig-zag tree and then an accepting run $\rho_t$ of $\+A$. 
Conversely, for every accepting run $\rho$ of $\+A$, there is an accepting run $t$ of $\+B$ such that $\rho_t=\rho$. 
Further, the cost computed by $\+B$ is closely related with $f$ in the sense that $\frac{1}{k}f(\rho_t)\leq \cost(t)\leq f(\rho_t)$\, ($\dagger$), for every accepting run $t$ of $\+B$, where
$k := |Q|^2 + |Q|$ is the maximum arity of a zig-zag tree. Then we have the following:
%
	\begin{align*}
	\cost_{\+B}(w) &\leq  \min \set{f(\rho) : \text{$\rho$ is an accepting run of $\+A$ on $w$}} \\
	&\leq  \min \set{\cost(\rho) : \text{$\rho$ is an accepting run of $\+A$ on $w$}} = \cost_{\+A}(w),
	\end{align*}
where $\min \emptyset = 0$.
Therefore, we have that if $\+A$ is limited, so is $\+B$. 

\smallskip

For the other direction, we claim that $\cost(\rho)\leq 3k^{f(\rho)}$, for every accepting run $\rho$ of $\+A$. 
To see this, consider the heaviest branch $B$ of $\rho$ (\ie, the result of traversing the tree from the root by always choosing a child whose subtree has maximal number of costly transitions). 
We can partition the edges of $B$ into $E_1$ and $E_2$ such that $E_1$ are the edges that do not decrease the cost of the current subtree and $E_2$ the ones that do. 
Since the initial cost is $n=\cost(\rho)$, 
and each edge in $E_2$ decreases the cost of the current subtree from $n'$ to no less than $\frac{n'}{k} -1$, 
we have $|E_2|\geq\max\{\ell\in \N: \frac{n}{k^\ell} - \sum_{i=0}^{\ell-1} \frac{1}{k^i} \geq 1\}\geq \max\{\ell\in \N: \frac{n}{k^\ell} - 2 \geq 1\}\geq \log_k{n/3}$. 
The claim follows since $f(\rho)\geq |E_2|$ (as each edge in $E_2$ is either costly or has a heavy branching). 

From the bound above, we can obtain that $\cost_{\+A}(w)\leq 3k^{k\cdot\cost_{\+B}(w)}$, for every word $w$, and hence if $\+B$ is limited, so is $\+A$. 
Indeed, take an accepting run $t$ of $\+B$ with $\cost(t)=\cost_{\+B}(w)$, 
and consider the associated accepting run $\rho_t$ of $\+A$. 
By ($\dagger$), we have $f(\rho_t)\leq k\cdot \cost(t)$. 
Summing up, we obtain $\cost_{\+A}(w)$ $\leq \cost(\rho_t)$ $\leq 3k^{f(\rho_t)}$ $\leq 3k^{k\cdot\cost_{\+B}(w)}$, as required. 
\end{proof}

\paragraph*{(3) From ADA$^\epsilon$ to ADA}
This is a straightforward polynomial time reduction. This reduction ---as opposed to reduction (1)--- does not preserve the language: we need to add an extra letter $a_\epsilon$ to the alphabet in order to make the reduction work in polynomial time. In fact, even for alternating finite automata (AFA) there is no known polynomial time language preserving translation from AFA with epsilon transitions into AFA (to the best of our knowledge). \sidediego{alguien conoce algún resultado al respecto?}
Given a ADA$^\epsilon$ $\+A = (\A, Q_{\exists}, Q_{\forall}, q_{0}, F, \delta)$, we can assume, without any loss of generality, that the state determines whether we are in a leftmost position, a rightmost position, or an internal position. That is, the statespace is partitioned into $Q_\exists = Q_{\exists,1} \dcup Q_{\exists,2} \dcup Q_{\exists,3}$ and $Q_\forall = Q_{\forall,1} \dcup Q_{\forall,2} \dcup Q_{\forall,3}$ so that $q_0  \in Q_{\exists,1}$ and every transition $(q,\alpha,e,c,p) \in \delta$ with $\alpha \in \A \cup \set\epsilon$, $q \in Q_{\exists,i} \cup Q_{\forall,i}$ and $p \in Q_{\exists,j} \cup Q_{\forall,j}$ is so that: (i) $i \leq j$, (ii) $e = \yend$ iff $j\in \set{1,3}$, (iii) if $\alpha=\epsilon$ then $i=j$.

We obtain $\+B = (\B, Q_{\exists}, Q_{\forall}, q_{0}, F, \delta')$ by extending the alphabet $\A$ with a new letter $\B = \A \cup \set{a_\epsilon}$, and obtaining the $\epsilon$-free transition relation $\delta'$ from $\delta$ by
\begin{itemize}
	\item replacing each transition $(q,\epsilon,e,c,p)$ with $(q,a_\epsilon,e,c,p)$, and
	\item adding self-loops $(q,a_\epsilon,\yend,0,q)$ for each state $q \in \bigcup_{\dag \in \set{\exists,\forall}, i \in \set{1,3}} Q_{\dag,i}$ and $(q,a_\epsilon,\nend,0,q)$ for each state $q \in Q_{\exists,2} \cup Q_{\forall,2}$.
\end{itemize}
It is easy to see that this reduction preserves limitedness. 

\paragraph*{(5) From ADA to DA}
Finally, the last reduction is exactly the same as the reduction A2DA to 2DA, observing that when applied to a ADA it yields a DA. It is worth noting that a limitedness preserving reduction in the context of \textsl{infinite words} has been proposed in \cite[Lemma~6]{KuperbergB11}, but it produces an automaton with an exponential set of states.\sidediego{check again}

\paragraph*{The resulting composition (1) + (2) + (3) + (4) + (5)}
Let $\+A_i = (\A_i, Q_{i,\exists}, Q_{i,\forall}, q_{i,0}, F_i, \delta_i)$ be a A2DA, for each $i \in \set{1,2,3,4}$ so that, starting with $\+A_1$, we get $\+A_2, \+A_3, \+A_4$ as the result from the reductions $\+A_1 \xrightarrow{(1)+(2)} \+A_2 \xrightarrow{(3)+(4)} \+A_3 \xrightarrow{(5)} \+A_4$ described before. We obtain the following properties. 
\begin{itemize}
	\item $\+A_4$ has a polynomial number of states. More precisely,
	$|Q_{4,\forall}| = |Q_{2,\forall}| = 0$, and $|Q_{4,\exists}| = |Q_{3,\exists} \cup Q_{3,\forall}| \leq poly(|Q_{2,\exists}|) = poly(|Q_{1,\exists} \cup Q_{1,\forall}|)$.
	\item $\A_4$ is (singly) exponential. Due to \eqref{eq:alphabet-blowup-redux}, $|\A_2| = |\A_1| \cdot |\delta_1|^{2 \cdot |Q_{1,\forall}|}$, $|\A_3| = |\A_2|+1$, and $|\A_4| = |\A_3| \cdot |\delta_3|^{2 \cdot |Q_{3,\exists} \cup Q_{3,\forall}|}$ again due to \eqref{eq:alphabet-blowup-redux}. Since $\delta_3$ is singly exponential in $\+A_1$ (since it is polynomial in $\A_3$ and $Q_{3,\exists}$), and $Q_{3,\exists}$ is polynomial in $\+A_1$, $\A$ is singly exponential in $\+A_1$.
\end{itemize}
As explained before, thanks to the fact that the bound of Theorem \ref{thm:da-pspace} depends only on the number of states and not on the size of the alphabet nor the transition set, this enables a \pspace procedure for testing for limitedness of A2DA, which concludes the proof of Theorem~\ref{thm:a2da-limitedness-pspace}.

 \section{Appendix to Section~\ref{sec:c2rpq-bound}}
 \renewcommand{\P}{\mathbb{P}}%
 We encode an instance of the {\em tiling problem}  
 following the ideas used for showing \expspace-hardness for CRPQ-containment \cite{CGLV00}. We reduce from the following $2^n$-tiling problem, which is \expspace-complete. An input instance consists of a number $n \in \N$ written in unary, a finite set $\Delta$ of tiles, two relations $H,V \subseteq \Delta \times \Delta$ specifying constraints on how tiles should be placed horizontally and vertically, and the starting and final tiles $t_S, t_F \in \Delta$. A solution to the input instance 
 is a ``consistent'' assignment of tiles to a finite rectangle having $2^n$ columns. Concretely, a solution is a function $f: \set{1, \dotsc, 2^n} \times \set{1, \dotsc, k} \to \Delta$, for some $k \in \N$, such that $f(1,1) = t_S$, $f(2^n,k)=t_F$, and 
 $f((i,j),f(i+1,j)) \in H$ and $f((i,j),f(i,j+1)) \in V$ for every $i,j$ in range.
 We can then obtain the following. 

 \begin{lemma*}[restatement of Lemma~\ref{lem:tiling-encoding}]
\lemTilingEncoding
 \end{lemma*}
 \begin{proof}
 For any tiling instance as above, we show how to define a CRPQ over the alphabet $\A := \Delta \cup \set{0,1,\#}$ so that it has at least $m$ and at most  $(g(|T|) + m)^{n+1}$ minimal expansions for some doubly-exponential function $g$, where $m$ is the number of solutions of the instance. We will encode a solution of a tiling as a word of $\#((0+1)^n \cdot \Delta)^*\#$, where the rectangle of tiles is read left-to-right and top-to-bottom, and each block (\ie, each element of $(0+1)^n ~ \Delta$) represents the column number (in binary) and the tile. The symbols $\#$ at the beginning and end of the word are used for technical reasons. 
 
 For enforcing this encoding, we define regular languages $E$, $F_C$, $F_H$ and $G_i$ for each $i \leq n$ over $\A$.

 The language $E$ gives the general shape of the encoding of solutions: 
 \[E= \# 0^n ~ t_S ~ ((0+1)^n  \Delta)^* ~ 1^n ~ t_F \#,\]
 in particular that it starts and ends with the correct tiles.
 The language $F_C$ detects adjacent blocks with an error in the column number bit, which can be easily defined with a polynomial NFA. The language $F_H$ checks that there are adjacent blocks in which the tiles do not respect the horizontal adjacency relation $H$, 
 \[F_H = \bigcup_{(t_1,t_2) \in \Delta^2 \setminus H} t_1 ~ \overline{0^n}  ~ t_2,  \]
  where $\overline{0^n} = (0+1)^n \setminus \set{ 0^n }$. 
 Finally, $G_0, \dotsc, G_n$ are used to check that there are two blocks at distance $2^n$ which do not respect the vertical adjacency relation $V$; in other words, there is a factor of the word whose first and last blocks have the same column number, it contains not more than one block with column number $1^n$ (otherwise we would be skipping a row), and its first and last tiles are not $V$-related.
 First, $G_0$ checks that the first and last blocks of the factor we are interested in do not conform to $V$, and furthermore that there is exactly one column number $1^n$ in between
 \[
 G_0 = \bigcup_{(t_1,t_2) \in \Delta^2 \setminus V} (0+1)^n t_1 (\overline{1^n} ~ \Delta)^* 1^n (\Delta \cdot \overline{1^n})^* t_2,
 \]
 where $\overline{1^n} = (0+1)^n \setminus \set{ 1^n }$. For each $b \in \set{0,1}$ and $i \in \set{1, \dotsc, n}$ we define $G_i^b$ to check that the $i$-th bit of the address of both the first and last tile is set to $b$,
 \[
   G_i^b = (0+1)^{i-1} \cdot b \cdot (0+1)^{n-i} \cdot \Delta \cdot ((0+1)^n \cdot \Delta)^* \cdot (0+1)^{i-1} b (0+1)^{n-i} \Delta,
 \]
 and we define $G_i$ as $G_i^0 + G_i^1$. For each one of these languages one can produce a regular expression recognizing the language in polynomial time.
 Finally, the Boolean CRPQ is
 \[
   \gamma = \exists x,y ~~ \bigwedge_{0 \leq i \leq n} x \xrightarrow{E \cup G_i \cup F_C \cup F_H} y.
 \]
 Let us analyse the bounds on the number of minimal expansions of $\gamma$ with respect to the number $m$ of solutions of the tiling problem.
 First, note that for any expansion of $\gamma$ containing a word $w \in E$ which \emph{does not} encode a solution to the tiling problem either: (i) it has a problem with the column encodings, in which case it contains a (polynomial) word from $F_C$ as a factor; 
 (ii) the encoding is correct, but the horizontal relation is not respected, in which case it contains a (polynomial) word from $F_H$ as a factor;  
 (iii) it violates the vertical relation, and thus there are words $w_i \in G_i$ for each $i$ so that $|w_i| \in O(n2^n)$ and the expansion corresponding to $w_0, \dotsc, w_n$ maps to the expansion.  
 Therefore, every path of a minimal expansion of $\gamma$ which is in $E$ and is not a solution cannot have size bigger than $O(n2^n)$, which means that the number of minimal expansions is at most $(|\A|^{O(n2^n)} + m)^{n+1}$. On the other hand, it follows by construction that every word encoding a solution is in $E$, that the expansion consisting of only solutions is minimal, and hence that there are at least $m$ minimal expansions of $\gamma$. 
\end{proof}

 Without loss of generality we can assume that the tiling instance satisfies that if there is a tiling solution, there are infinitely many. This fact, coupled with Lemma~\ref{lem:tiling-encoding} and \expspace-completeness of the $2^n$-tiling problem, yields the lower bound in part (1) of Theorem
 \ref{theo:main}.

 \begin{proposition}\label{prop:bounded-crpq-expspace-hard}
 {\sc Boundedness} for CRPQs is \expspace-hard. This holds even for Boolean CRPQs of the form $\exists x,y \, \bigwedge_i (x \xrightarrow{L_i} y)$, whose languages $L_i$ are given as regular expressions.
 \end{proposition}

 \smallskip
 \noindent 
 {\bf {\em Lower bound on size of equivalent UCQ.}} 
 It is not hard to produce $2^n$-tiling instances $T_n$ having triple-exponentially many solutions. Indeed, it suffices to (1) enforce that each solution has exactly $2^{2^n}$ rows (and hence there are only finitely many solutions), which can be done by encoding the binary representation of $i$ at each row $i$, and (2) encode at each row an arbitrary symbol from the alphabet $\set{a,b}$. In this way, each solution encodes a function $f : R \to \set{a,b}$, where 
 $R = \set{0, \dotsc, 2^{2^n}-1}$, and conversely, for each such a function there is a distinct solution. It then follows that $T_n$ has $2^{2^{2^n}}$ solutions. 
 In particular, the Boolean CRPQ $\gamma_n$ from Lemma~\ref{lem:tiling-encoding} is bounded and has at least $2^{2^{2^n}}$ minimal expansions. 
Recall that these minimal expansions are produced by expanding each atom of $\gamma_n$ into a word $w\in E$ corresponding to a solution of $T_n$. 
Hence, if $\lambda$ and $\lambda'$ are two of these minimal expansions, we have that $\lambda\not\to \lambda'$, \ie, $\gamma_n$ has at least $2^{2^{2^n}}$ homomorphically incomparable minimal expansions. 
By Lemma~\ref{lemma:cont-crpq}, it follows that every UCQ equivalent to $\gamma_n$ must have at least $2^{2^{2^n}}$ disjuncts.  
 This yields part (3) of Theorem \ref{theo:main}.

\section{Appendix to Section~\ref{sec:nice-c2rpq}}

\paragraph*{Proof of Theorem \ref{theo:strongly-acyclic}}

The \pspace lower bound follows from Corollary~\ref{coro:comp-rpq}, so we focus on the upper bound. 
Given an acyclic UC2RPQ $\Gamma$ of thickness $\leq k$, we shall construct in polynomial time an A2DA$^\epsilon$ $\+A$ 
of polynomial size in $\|\Gamma\|$ such that $\Gamma$ is bounded iff $\+A$ is limited. The result will follow from Theorem~\ref{thm:a2da-limitedness-pspace}. 

As in the proof of the \expspace upper bound in Theorem~\ref{theo:main}, the A2DA$^\epsilon$ $\+A$ will run over encodings of expansions of $\Gamma$. 
So if $\A$ is the alphabet of $\Gamma$, then the alphabet of $\+A$ is $\A_1:=\A^{\pm} \cup {\cal V} \cup \{\$\}$, where ${\cal V}$ is the set of variables of $\Gamma$ and $\$$ is a fresh symbol. 
Again, if $\lambda$ is the expansion of a disjunct $\gamma=\exists \bar z \bigwedge_{1 \leq i \leq m} (x_i \xrightarrow{L_i} y_i)$ of $\Gamma$ obtained by expanding
$x_i \xrightarrow{L_i} y_i$ into an oriented path $\pi_i$ from $x_i$ to $y_i$ with label $w_i\in L_i$, then we encode $\lambda$ as the word over $\A_1$
$$w_{\lambda} = \$ x_1 w_1 y_1 \$ x_2 w_2 y_2 \$ \, \dotsb \, \$ x_m w_m y_m \$ $$

Note how the subword $x_iw_iy_i$ represents the oriented path $\pi_i$. 
Every position $j\in \{1,\dots,|w_\lambda|\}$ with $w_\lambda[j]\neq \$$ represents a variable in $\lambda$: either $x_i$ or $y_i$ if $w_\lambda[j]=x_i$ or $w_\lambda[j]=y_i$, respectively; 
or the $(\ell+1)$-th variable in the oriented path $\pi_i$ if $w_\lambda[j]$ is the $\ell$-th symbol in the subword $w_i$. 
Hence different positions could represent the same variable in $\lambda$: \eg, in the encoding $\$xabcy\$$, the 5\emph{th} position containing a `$c$' and 6\emph{th} position containing a `$y$', represent the same variable, namely, the last vertex $y$ of the oriented path.

It follows from the definition of $\+B_2^\gamma$ and $\+D_2^\gamma$ in Section~\ref{sec:upper-bound} that for every regular language $L$ appearing in $\Gamma$, 
there is a 2DA $\+A_L$ over $\A_1$, computable in polynomial time, such that for every expansion $\lambda$ of $\Gamma$, and head positions $i,j\in \{1,\dots,|w_\lambda|\}$, 
where  $x_i$ and $x_j$ are the variables in $\lambda$ 
represented by $i$ and $j$, respectively, we have:

\begin{itemize}
\item Every accepting run $\rho$ of $\+A_L$ over $w_\lambda$ from position $i$ to position $j$, 
determines an oriented path $\pi_\rho$ in $\lambda$ from $x_i$ to $x_j$ whose label is in $L$.  Moreover, $\cost(\rho)$ is precisely the number of atoms of $\pi_\rho$. 
\item For every oriented path $\pi$ in $\lambda$ from $x_i$ to $x_j$ with label in $L$, there is an accepting run $\rho$ of $\+A_L$ over $w_\lambda$ from position $i$ to position $j$, 
such that $\pi_\rho=\pi$. 
\end{itemize}




Recall also from Section~\ref{sec:upper-bound} that there is an NFA $\+A_1$ that accepts precisely those words over $\A_1$ that encode some 
expansion of $\Gamma$.  The size of $\+A_1$ is polynomial in $\|\Gamma\|$. 
From $\+A_1$ we can obtain a 2NFA $\+C_1$ that, 
over a word $w$, starts by executing $\+A_1$ until we reach the position $|w|$ and if we reach a final state of $\+A_1$ then we move to the position $0$ and accept, 
\ie, we enter the final state of $\+C_1$. 
Let $\+D_1$ be the DA obtained from $\+C_1$ by setting the cost of all transitions to be $0$. 
Our A2DA$^\epsilon$ $\+A$ is the concatenation of $\+D_1$ and $\+B$ (defined below), \ie, we add non-costly $\epsilon$-transitions from the unique final state of $\+D_1$ to the initial state of $\+B$.
Note that a word that is not of the form $w_\lambda$ for some expansion $\lambda$ of $\Gamma$ has no accepting run of $\+A$. 
Hence $\+A$ is limited iff it is limited over the words $w_\lambda$'s. 
Moreover, we have $\cost_{\+A}(w_\lambda)=\cost_{\+B}(w_\lambda)$. 
Thus, $\+A$ is limited iff $\+B$ is limited over all words of the form $w_\lambda$. 
We shall construct $\+B$ so the latter condition is equivalent to $\Gamma$ being bounded. 
In particular, for an expansion $\lambda$ of $\Gamma$, let $\lambda_{min}$ be a minimal size expansion of $\Gamma$ 
such that $\lambda_{min}\to \lambda$. We will show that 
\begin{align}
\cost_{\+B}(w_\lambda) \leq \|\lambda_{min}\| \leq g(\cost_{\+B}(w_\lambda)), \, \text{ for every expansion $\lambda$}, \label{eq:goal} 
\end{align}
where $g$ is some non-decreasing function. 
Hence, $\+B$ is limited over all words of the form $w_\lambda$ iff $\|\lambda_{min}\|$ is bounded over all expansions $\lambda$ of $\Gamma$. 
The latter condition is equivalent to boundedness of $\Gamma$ by our characterization of Proposition~\ref{prop:basic}, item (2).

Before defining $\+B$ we need to introduce some notation. 
Let $\gamma$ be a disjunct of $\Gamma$. 
A \emph{connected component} of $\gamma$ is a maximal subquery whose underlying graph is connected. 
Since $\gamma$ is acyclic, 
we can assume that (the underlying graph of) every connected component of $\gamma$ is a rooted tree, 
and we use the usual terminology of trees (parent, children, leaves, $\dots$) over the variables of $\gamma$.
For variables $x,y$ in $\gamma$, we define $\text{Atoms}_\gamma(x,y)$ to be 
the set of atoms of $\gamma$ of the form $x \xrightarrow{L} y$ or $y \xrightarrow{L} x$. 
Suppose $x$ is the parent of $y$ in $\gamma$. 
Without loss of generality, we shall assume that each atom in $\text{Atoms}_\gamma(x,y)$ is of the form $x \xrightarrow{L} y$ (otherwise, we simply ``reverse'' $L$). 
We also assume a fixed enumeration $L_1,\dots, L_\ell$, where $\ell=|\text{Atoms}_\gamma(x,y)|$, of the regular languages labelling the atoms of $\text{Atoms}_\gamma(x,y)$. 
We define $\text{Cuts}_\gamma(x,y)$ to be the set of \emph{cuts} from $x$ to $y$, that is, the set of tuples $(q_1,\dots,q_\ell)$, where 
each $q_i$ is a state of $\+A_{L_i}$. We say that $(q_1,\dots,q_\ell)$ is an \emph{initial cut} if each $q_i$ is the initial state of $\+A_{L_i}$. 
Similarly, we say that $(q_1,\dots,q_\ell)$ is a \emph{final cut} if each $q_i$ is a final state of $\+A_{L_i}$. 
We also define $\text{Cuts}_\gamma=\bigcup\{\text{Cuts}_\gamma(x,y): \text{$x$ parent of $y$ in $\gamma$}\}$ and $\text{Triples}_\gamma=\{(L,q,q'): \text{$L$ appears in $\gamma$ and $q,q'$ states in $\+A_L$}\}$. 

The main idea of the A2DA$^\epsilon$ $\+B$ is similar to the idea behind Section~\ref{sec:upper-bound}: 
for an encoding $w_\lambda$ of an expansion $\lambda$ of $\Gamma$, the automaton $\+B$ tries to map some expansion $\lambda'$ of $\Gamma$ into (the encoding of) $\lambda$. 
In particular, every accepting run $\rho$ of $\+B$ over $w_\lambda$ determines an expansion $\lambda_\rho$ such that $\lambda_\rho\to \lambda$. 
On the other hand, for every expansion $\lambda'$ such that $\lambda'\to \lambda$ there is an accepting run $\rho$ of $\+B$ over $w_\lambda$ with $\lambda_\rho=\lambda'$. 
We will prove that 
\begin{align}
\cost(\rho) \leq \|\lambda_{\rho}\| \leq g(\cost(\rho)), \, \text{ for every accepting run $\rho$ of $\+B$ over $w_\lambda$}, \label{eq:goal2} 
\end{align}
where $g$ is some non-decreasing function. Note then that (\ref{eq:goal2}) implies (\ref{eq:goal}). 

The main difference from the construction of Section~\ref{sec:upper-bound} is that we do not need to annotate first the encoding $w_\lambda$ and then project the annotations. 
Instead, we can exploit the acyclicity of $\Gamma$ to try to map directly an expansion of $\Gamma$ to the input expansion $\lambda$. 
To do this, $\+B$ starts by choosing a disjunct $\gamma$ of $\Gamma$. Then $\+B$ applies universal transitions to map all the connected components of $\gamma$. 
Each component is mapped in a top-down fashion starting from the root to the leaves. 
Once some variable $x$ is already mapped to some variable $h(x)$ of $\lambda$, or more precisely, 
to some head position $j_x\in \{0,\dots,|w_\lambda|\}$ such that $j_x>0$ and $w_\lambda[j_x]$ represents the variable $h(x)$,
 then $\+B$ applies universal transitions to what we call the \emph{axes} of $x$. 
An axis of $x$ is either an atom $x\xrightarrow{L}x\in \text{Atoms}_\gamma(x,x)$ or a child of $x$ in $\gamma$. 
We need to extend our mapping to every axis of $x$. 

By definition of the 2DA $\+A_L$, the mappings of an axis $x\xrightarrow{L}x\in \text{Atoms}_\gamma(x,x)$ correspond to the accepting runs of $\+A_L$ over $w_\lambda$ 
starting and ending at position $j_x$. In order to find these looping accepting runs of $\+A_L$, we use a similar idea as in the reduction (3) in Theorem~\ref{thm:a2da-limitedness-pspace} 
from 2DA to ADA$^{\epsilon}$. We have that every accepting run  of $\+A_L$ over $w_\lambda$ 
starting and ending at position $j_x$, can be divided into a set of rightward looping partial runs and a set of leftward looping partial runs. 
Following the terminology of reduction (3), we can represent rightward and leftward looping partial runs by \emph{rightward} and \emph{leftward zig-zag trees}, respectively, 
which are rooted trees of height at most $|w_\lambda|$ where each edge has a cost in $\{0,1,2\}$ and each node is labeled with $(L,q,q')$ for a pair of states $q,q'$ in $\+A_L$. 
In a rightward \resp{leftward} zig-zag tree, each level $\ell$ from the root to the leaves corresponds to the position $j_x+\ell$ in $\{0,\dots,|w_\lambda|\}$ \resp{$j_x-\ell$}; \cf, Figure~\ref{fig:2way-as-tree}. 
If $u$ has only one child $v$ in a rightward zig-zag tree (the leftward case is analogous) and their labels are $(L,q,p)$ and $(L,q',p')$, respectively, 
then there must be transitions of the form $(q,a,c_1,q')$ and $(p',a^{-1},c_2,p)$ in $\+A_L$ such that $a=w_\lambda[j_x+\ell_v]$, where $\ell_v$ is the level of $v$ (\ie, its distance to the root). 
The cost of the edge $\{u,v\}$ is then $c_1+c_2$. 
If $u$ has children $v_1,\dots,v_r$, with $r\geq 2$, and the label of $u$ is $(L,q,p)$ then the labels of $v_1,\dots,v_r$ must be $(L,q_0,q_1), (L,q_1,q_2),\dots,(L,q_{r-1},q_r)$, respectively, 
and there must be transitions $(q,a,c_1,q_0)$ and $(q_r,a^{-1},c_2,p)$ of $\+A_L$ such that $a=w_\lambda[j_x+\ell]$, where $\ell$ is the level of the $v_i$'s. 
The cost of $\{u,v_i\}$ is $c_1$ for $i=1$, $c_2$ for $i=r$, and $0$ for $1<i<r$. 
Finally, every leaf in a rightward or leftward zig-zag tree has a label of the form $(L,q,q)$, for some state $q$ in $\+A_L$.

In order to map $x\xrightarrow{L}x$, the automaton $\+B$ chooses a final state $q'$ of $\+A_L$ and enters a state $(L,q,q')^\epsilon$, where $q$ is the initial state of $\+A_L$. 
From there, $\+B$ can spawn threads starting from states $(L,q,p,s)$  and $(L,p,q')^\epsilon$, where $s\in \{\text{right},\text{left}\}$ and the state $(L,q,p,s)$ indicates 
that we are looking for a $s$-ward partial run $\rho$ from $q$ to $p$ starting and ending at $j_x$. As in the reduction (3), 
this is done by exploiting alternation to guess the zig-zag tree of $\rho$ in a top-down manner. 
In order to compute the cost of an accepting run correctly, as in reduction (3), the states of $\+B$ of the form $(L,q,p,s)$ are enhanced with a number $c\in \{0,1\}$ (and denoted by $[(L,q,p,s)]_c$), which 
indicates whether there is a costly transition in the looping partial run from $q$ to $p$ (or equivalently, in the zig-zag subtree rooted at $(L,q,p)$). 

For an axis corresponding to a child $y$ of $x$, we need to map simultaneously all the atoms $x\xrightarrow{L_1}y$, $\dots$, $x\xrightarrow{L_\ell}y$ $\in\text{Atoms}_\gamma(x,y)$. 
The idea is first to choose whether $y$ is mapped to the right or to the left of $j_x$, and then moving in the chosen direction guessing a sequence of cuts $D_0,\dots,D_n\in \text{Cuts}_\gamma(x,y)$, 
where $D_0$ is an initial cut and $D_n$ is final. These are represented by $\+B$ via states $(D_i,s)$, where $s\in \{\text{right},\text{left}\}$. 
A transition from $D_i$ to $D_{i+1}$ can either consume a symbol from $w_\lambda$ or be an $\epsilon$-transition 
that modifies only one coordinate of $D_i$, say $j$, and produces $D_{i+1}$. When the latter happen, $\+B$ spawn threads starting in states $(D_{i+1},s)$ and $(L_j,q,q',s')$, 
where $q$ and $q'$ are the $j$-th coordinates of $D_i$ and $D_{i+1}$, respectively, and $s'\in \{\text{right},\text{left}\}$ is some direction. 
The intuition is that we choose to do an asynchronous mapping only for $\+A_{L_j}$ in the form of a $s'$-ward looping partial run from $q$ to $q'$. 
This is needed as the mappings of the atoms $x\xrightarrow{L_j}y$ into $w_\lambda$ can be very different from each other. 

Observe that we can represent mappings of $\{x\xrightarrow{L_1}y$, $\dots$, $x\xrightarrow{L_\ell}y\}$ $=\text{Atoms}_\gamma(x,y)$ into $w_\lambda$, where $x$ and $y$ are mapped to $j_x$ and $j_y$ 
(we assume $j_y\geq j_x$; the case $j_y\leq j_x$ is analogous), respectively, 
 as \emph{rightward special trees} from $j_x$ to $j_y$ (or leftward in the case $j_y\leq j_x$), 
which are rooted trees with costs in the edges obtained from a
special rooted path $B$ with $j_y-j_x+1$ nodes, each of which is labeled with a pair of cuts $(D,D')$ in $\text{Cuts}_\gamma(x,y)$, by attaching some zig-zag trees to each node as explained below. 
If $(D,D')$ is the label of the root, \ie, the first node of $B$, then $D$ is an initial cut, 
and if $(D,D')$ is the label of the last node of $B$ then $D'$ is a final cut. If $\{u,v\}$ is the $r$-th edge of $B$ with $r\in \{1,\dots,j_y-j_x\}$, and $(D,D')$ and $(E,E')$ are the labels of $u$ and $v$ respectively, 
then for every coordinate $i\in \{1,\dots,\ell\}$, there is a transition in $\+A_{L_i}$ of the form $(q,a,c_i,q')$, where $q$ and $q'$ are the $i$-th coordinates of $D'$ and $E$, respectively, 
and $a=w_\lambda[j_x+r]$. The cost of the edge $\{u,v\}$ is $c_1+\dots+c_\ell\in \{0,\dots,\ell\}$. 
In a rightward special tree we also have that each node $u$ in the special rooted path $B$ is associated with a set of disjoint rightward or leftward zig-zag trees whose roots have been identified with $u$. 
In particular, if $u$ has label $(D,D')$ and $q_i$ and $q'_i$ are the $i$-th coordinates of $D$ and $D'$, respectively, then for every $i\in \{1,\dots, \ell\}$ there is a sequence 
$p_i^0, p_i^1,\dots,p_i^{n_i}$ such that $p_i^0=q_i$, $p_i^{n_i}=q'_i$ and the set of labels of the roots of the rightward or leftward zig-zag tree associated with $u$ (before being identified with $u$) 
is precisely $\{(L_i,p_i^j, p_i^{j+1}): i\in \{1,\dots,\ell\}, j\in \{0,\dots,n_i-1\}\}$. 
Observe that the working of $\+B$ explained in the previous paragraph for mapping $x\xrightarrow{L_1}y$, $\dots$, $x\xrightarrow{L_\ell}y$ $\in\text{Atoms}_\gamma(x,y)$ 
into $w_\lambda$ can be seen as using the power of alternation to guess a rightward or leftward special tree from the current position $j_x$ to some position $j_y$. 
Again, in order to compute the cost correctly, we need to consider states of the form $[(D,s)]_c$, for $c\in\{0,1\}$, instead of $(D,s)$. 

Now we are ready to formally define the A2DA$^\epsilon$ $\+B$. We shall use transitions of the form $(q,a,c,q')$ where $c\in \{0,\dots, k\}$.
Note that this is not a problem as they can be simulated using $\epsilon$-transitions. 
Moreover, the automaton $\+B$ will not need to move its head to the leftmost and rightmost positions (as these never represent a variable of the input expansion), so almost all of its transitions will be $\nend$-flagged. 
Hence, for a transition, we will write $(q,a,c,q')$ instead of $(q,a,\nend,c,q')$. 
We define $\+B=(\A_1, Q_\exists, Q_\forall, q_0,F,\delta)$, where 
\begin{itemize}
\item $Q_\exists=\{q_0, q_f\}\,\cup\,  \bigcup\{Q_\exists^\gamma: \text{$\gamma$ disjunct of $\Gamma$}\}$. 
For a disjunct $\gamma$ of $\Gamma$, we define 
$$Q_\exists^\gamma=((\text{Cuts}_\gamma\cup \text{Triples}_\gamma)\times\{\text{right,left}\} \cup \text{Triples}_\gamma^{\epsilon})\times\{0,1\}\, \cup {\text{Axes}}_\gamma \cup \text{Roots}_\gamma,$$ 
where $\text{Triples}_\gamma^{\epsilon}:=\{t^\epsilon: t\in \text{Triples}_\gamma\}$, 
$${\text{Axes}}_\gamma:=\{x_y: \text{$x$ is the parent of $y$ in $\gamma$}\}\cup\{x_{A}: \text{$x$ is in $\gamma$, $A\in \text{Atoms}_{\gamma}(x,x)$}\},$$
and $\text{Roots}_\gamma:=\{r_{\text{init}}: \text{$r$ is the root of some connected component of $\gamma$}\}$.  
We shall write $[\alpha]_c$ for $(\alpha,c)\in Q_\exists^\gamma\setminus ({\text{Axes}}_\gamma \cup \text{Roots}_\gamma)$. 
\item $Q_\forall=\bigcup\{\{q_0^\gamma\}\, \cup\, Q_\forall^\gamma: \text{$\gamma$ disjunct of $\Gamma$}\}\,\cup\,\{q_{\text{nf}}\}$. 
For a disjunct $\gamma$ of $\Gamma$, we define 
$Q_\forall^\gamma={\cal V}_\gamma \cup \{[\alpha_1]_{c_1}\land [\alpha_2]_{c_2}: [\alpha_1]_{c_1}\in Q_\exists^\gamma\setminus ({\text{Axes}}_\gamma \cup \text{Roots}_\gamma), [\alpha_2]_{c_2}\in \text{Triples}_\gamma\times \{\text{right},\text{left}\}\times\{0,1\} \}$. 
\item $F=\{q_f\}\cup \bigcup\{[(L,q,q,s)]_0, [(L,q,q)^\epsilon]_0 : \text{$\gamma$ disjunct of $\Gamma$}, (L,q,q)\in \text{Triples}_\gamma, s\in\{\text{right},\text{left}\}\}$. 
\item $\delta$ is the smallest set verifying:
\begin{enumerate}
\item $(q_0,\epsilon, \yend,0, q_0^\gamma)\in \delta$, for every disjunct $\gamma$ of $\Gamma$, so we can choose the disjunct of $\Gamma$ to be mapped into the input expansion. 
\item $(q_0^\gamma, \epsilon, \yend,0, r_{\text{init}})\in \delta$, for every disjunct $\gamma$ of $\Gamma$ and every root $r$ of a connected component of $\gamma$. 
For the disjunct $\gamma$ to be mapped, we need all of its connected components to be mapped. 
\item For every disjunct $\gamma$ of $\Gamma$, every root $r$ of a connected component of $\gamma$, and every $a\in \A_1^{\pm}$, we have
 $\{(r_{\text{init}}, a, 0, r_{\text{init}}), (r_{\text{init}}, \epsilon, 0, r)\}\subseteq \delta$. With these transitions we can choose the position where the root $r$ should be mapped. 
\item $(x,\epsilon, 0, x_t)\in \delta$, for every disjunct $\gamma$ of $\Gamma$, every $x$ in $\gamma$ and every $x_t\in \text{Axes}_\gamma$. Once we mapped $x$ into the expansion, 
we need to map all of its subtrees and all atoms in $\text{Atoms}_\gamma(x,x)$. 
\item For every $\gamma$ in $\Gamma$, and every free variable $x$ in $\gamma$, we have $(x,x^{-1}, 0, q_f)\in \delta$ and $(x,b^{-1}, 0, q_{\text{nf}})\in \delta$, for every $b\in \A_1\setminus \{x\}$. This ensures that $x$ is always mapped to itself. 
\item For every $\gamma$ in $\Gamma$ and $x_t\in \text{Axes}_\gamma$, 
\begin{itemize}
\item If $t=y$ (and hence $x$ is the parent of $y$ in $\gamma$), then $(x_t, \epsilon, 0, [(D,s)]_c)\in \delta$, for every $c\in \{0,1\}$, c and every initial cut $D\in \text{Cuts}_\gamma(x,y)$.
\item If $t=x \xrightarrow{L} x$, then $(x_t, \epsilon, 0, [(L,q,q')^\epsilon]_c)\in \delta$, where $q$ is the initial state of $\+A_L$, every final state $q'$ of $\+A_L$, and every $c\in \{0,1\}$. 
\end{itemize}
These transitions allow us to start looking for a mapping of each axis of $x$ into the expansion.   
\item For every $\gamma$ in $\Gamma$, 
\begin{itemize}
\item $([(D,s)]_{\max(c_1,c_2)}, \epsilon, 0, [(D',s)]_{c_1}\land [(L,q,q',s')]_{c_2})\in \delta$, for every $c_1,c_2\in \{0,1\}$, $s,s'\in \{\text{right},\text{left}\}$ and every $D,D'\in \text{Cuts}_\gamma(x,y)$, for some $x,y$, 
such that $D=(q_1,\dots,q_\ell)$, $D'=(q'_1,\dots,q'_\ell)$, there is $j$ such that $q_i=q'_i$, for all $i\in \{1,\dots, \ell\}\setminus\{j\}$, and 
$(L,q,q')=(L_j,q_j,q_j')$, where $L_j$ is the $j$-th language mentioned in $\text{Atoms}_\gamma(x,y)$.  With these transitions we guess that a looping $s'$-ward partial run of $\+A_L$ from $q$ to $q'$ should be mapped to the input expansion. 
In terms of special trees, these transitions allow us to add new subtrees to a node in the special path of the $s$-ward special tree. 
\item  $([(L,q,q')^\epsilon]_{\max(c_1,c_2)}, \epsilon, 0, [(L,p,q')^\epsilon]_{c_1}\land [(L,q,p,s)]_{c_2})\in \delta$, for every $c_1,c_2\in \{0,1\}$, $s\in\{\text{right},\text{left}\}$ 
and $q,q',p$ states in $\+A_L$. With these transitions we can guess that a looping $s$-ward partial run of $\+A_L$ from $q$ to $p$ should be mapped to the input.   
\item  $([(L,q,q',s)]_{\max(c_1,c_2)}, \epsilon, 0, [(L,q,p,s)]_{c_1}\land [(L,p,q',s)]_{c_2})\in \delta$, for every $c_1,c_2\in \{0,1\}$, $s\in\{\text{right},\text{left}\}$ and $q,q',p$ states in $\+A_L$. 
We reduce the search for the looping $s$-ward partial run of $\+A_L$ from $q$ to $q'$, to look for loopings $s$-ward partial runs from $q$ to $p$ and from $p$ to $q'$. 
In terms of zig-zag trees, these transitions allow us to add a new subtree to the $s$-ward zig zag tree. 
\end{itemize}
\item For every $\gamma$ in $\Gamma$, we have $\{([\alpha_1]_{c_1}\land [\alpha_2]_{c_2}, \epsilon, c_2, [\alpha_1]_{c_1}), ([\alpha_1]_{c_1}\land [\alpha_2]_{c_2}, \epsilon, c_1, [\alpha_2]_{c_2})\}\subseteq \delta$, 
for every $[\alpha_1]_{c_1}\in Q_\exists^\gamma\setminus ({\text{Axes}}_\gamma \cup \text{Roots}_\gamma)$ and $[\alpha_2]_{c_2}\in \text{Triples}_\gamma\times \{\text{right},\text{left}\}\times\{0,1\}$. 
\item For every $\gamma$ in $\Gamma$, $a\in \A_1^\pm=\A_1\dcup\A_1^{-1}$, $c\in \{0,1\}$ and $s\in\{\text{right},\text{left}\}$, 
\begin{itemize}
\item $([(L,q,p,s)]_{\max\{c,d\}}, a, d, [(L,q',p',s)]_c)\in \delta$, if there exist transitions $(q,a,c_1,q')$ and $(p',a^{-1},c_2,p)$ in $\+A_L$ such that $d=\max\{c_1,c_2\}$; and $s=\text{right} \Leftrightarrow a\in \A_1$.  
We reduce the search for the looping $s$-ward partial run of $\+A_L$ from $q$ to $p$, to look for a looping $s$-ward partial run from $q'$ to $p'$. 
\item For every $D=(q_1,\dots,q_\ell)$ and $D'=(q'_1,\dots,q'_\ell)$ in $\text{Cuts}_\gamma(x,y)$, for some variables $x,y$, 
 we have $([(D,s)]_{\max\{c,d\}}, a, d, [(D',s)]_c)\in \delta$, if there are transitions $(q_1,a,c_1,q_1')$, $\dots$, $(q_\ell,a,c_\ell,q'_\ell)$ in $\+A_{L_1}, \dots, \+A_{L_\ell}$, respectively, 
 where $L_i$ is the $i$-th language mentioned in $\text{Atoms}_\gamma(x,y)$, such that $d=\max\{c_1,\dots,c_\ell\}$; and $s=\text{right} \Leftrightarrow a\in \A_1$.  
 These transitions correspond to a simultaneous mapping of all the atoms in $\text{Atoms}_\gamma(x,y)$ to the input expansion. 
 In terms of special trees, these correspond to traverse one edge of the special path. 
\end{itemize}
\item For $\gamma$ in $\Gamma$ and every final cut $D\in \text{Cuts}_\gamma(x,y)$, for some variables $x,y$ in $\gamma$, and $s\in\{\text{right},\text{left}\}$, 
we have $([(D,s)]_0, \epsilon, 0, y)\in \delta$. 
\item Finally, $(x,\epsilon, 0, q_f)\in \delta$, for every $\gamma$ in $\Gamma$ and $x$ an existentially quantified variable that is a leaf in $\gamma$ with $\text{Atoms}_\gamma(x,x)=\emptyset$.   
\end{enumerate}
\end{itemize}

\begin{lemma}
$\Gamma$ is bounded iff $\+B$ is limited.
\end{lemma}

\begin{proof}
Note first that every accepting run $\rho$ of $\+B$ on $w_\lambda$, for an expansion $\lambda$ of $\Gamma$, determines a disjunct $\gamma_\rho$ of $\Gamma$, an expansion $\lambda_\rho$ of $\gamma_\rho$ 
and a homomorphism $h_\rho$ witnessing $\lambda_\rho\to \lambda$. 
Conversely, for every disjunct $\gamma'$ of $\Gamma$, every expansion $\lambda'$ of $\gamma'$ and homomorphism $h'$
witnessing $\lambda'\to \lambda$, there is an accepting run $\rho$ of $\+B$ on $w_\lambda$ with $\gamma_\rho=\gamma'$, $\lambda_\rho=\lambda'$ and $h_\rho=h'$. 
Hence, it suffices to show the above-mentioned condition (\ref{eq:goal2}), \ie, 
$$\cost(\rho) \leq \|\lambda_{\rho}\| \leq g(\cost(\rho)), \, \text{ for every accepting run $\rho$ of $\+B$ over $w_\lambda$},\qquad\qquad (\ref{eq:goal2})$$
for some non-decreasing function $g$.
In order to show this, we follow an argument similar to the one of Lemma~\ref{lem:}. 
Note that every accepting run $\rho$ of $\+B$ on $w_\lambda$, determines a collection $C_\rho=Z_\rho\cup S_\rho$, 
such that $Z_\rho=\bigcup\{A_\rho: \text{$A\in \text{Atoms}_{\gamma_\rho}(x,x)$ for some variable $x$ in $\gamma_\rho$}\}$, 
where $A_\rho$ is a collection of rightward and leftward zig-zag trees; and $S_\rho=\{t_\rho^{x,y}: \text{$x$ is the parent of $y$ in $\gamma_\rho$}\}$, 
where $t_\rho^{x,y}$ is either a rightward or leftward special tree.

Let $t$ be a (rightward or leftward) zig-zag or special tree. We write $\cost(t)$ for the sum of all the costs, over all edges of $t$.  
For a branch $B$ of $t$, a \emph{heavy branching} of $B$ is a subtree attached to $B$ that has at least one edge with cost $>0$. 
We define $f(t)$ to be the maximum over all branches $B$ of $t$, of the cost of $B$ (\ie, the sum of the costs of the edges of $B$) plus the number of heavy branchings of $B$. 
Note that $f(t)\leq \cost(t)$. 
Now let $\rho$ be an accepting run of $\+B$ over $w_\lambda$. 
We have that $\|\lambda_\rho\|=\sum_{t\in C_\rho} \cost(t)$. 
By construction of $\+B$, we have that $\cost(\rho)\leq \sum_{t\in C_\rho} f(t)$. Then, 
$$ \cost(\rho)\leq \sum_{t\in C_\rho} f(t)\leq \sum_{t\in C_\rho} \cost(t)=\|\lambda_\rho\|,$$
which proves one of the directions of condition (\ref{eq:goal2}). 

For the other direction, 
for every $t\in C_\rho$, we have that $f(t)\leq r\cdot \cost(\rho)$, where $r:=1+\sum\{|Q_L|^4:\text{$L$ appearing in $\Gamma$}\}$, where $Q_L$ is the statespace of $\+A_L$ and $r$ is an upper bound for the maximum arity of 
any zig-zag or special tree. Also, for every $t\in C_\rho$, we have $\cost(t)\leq (2k+3)r^{f(t)}$. (Recall that $k$ is the upper bound on the thickness of $\Gamma$.)
Indeed, consider the heaviest branch $B$ of $t$ (\ie, the result of traversing $t$ from the root by always choosing a child whose subtree has maximal total cost). 
We can partition the edges of $B$ into $E_1$ and $E_2$ such that $E_1$ are the edges that do not decrease the total cost of the current subtree and $E_2$ the ones that do. 
Let $n:=\cost(t)$ be the initial total cost of $t$. 
Note that each edge in $E_2$ decreases the total cost of the current subtree from $n'$ to no less than $\frac{n'}{r}-(k+1)$ (note that $k+1$ is an upper bound for the cost of any edge in any zig-zag or special tree).  
 We have that $|E_2|\geq\max\{\ell\in \N: \frac{n}{r^\ell} - (k+1)\sum_{i=0}^{\ell-1} \frac{1}{r^i} \geq 1\}\geq \max\{\ell\in \N: \frac{n}{r^\ell} - 2(k+1) \geq 1\}\geq \log_r{\frac{n}{2(k+1)+1}}$. 
The claim follows since $f(t)\geq |E_2|$ (as each edge in $E_2$ either has cost $>0$ or has a heavy branching). 
Summing up, we have that 
$$\|\lambda_\rho\|=\sum_{t\in C_\rho} \cost(t)\leq \sum_{t\in C_\rho} (2k+3)r^{f(t)}\leq \sum_{t\in C_\rho} (2k+3)r^{r\cdot \cost(\rho)}\leq r\cdot N_\Gamma (2k+3)r^{r\cdot \cost(\rho)},$$
where $N_\Gamma$ is the number of atoms of $\Gamma$ (note that $r\cdot N_\Gamma$ is then an upper bound to $|C_\rho|$). 
This shows the remaining direction of condition (\ref{eq:goal2}), and hence the lemma. 
\end{proof}

Finally, note that the number of states of $\+B$ is polynomial in $\|\Gamma\|$, and hence $\+B$ can be constructed in polynomial time. 
Indeed, the crucial part is to bound $|\text{Cuts}_\gamma(x,y)|$ for any disjunct $\gamma$ and variables $x,y$ in $\gamma$. 
Since the thickness of $\Gamma$ is $\leq k$, we have that $|\text{Cuts}_\gamma(x,y)|\leq \|\Gamma\|^k$. This finishes the proof of Theorem~\ref{theo:strongly-acyclic}.

\paragraph*{Proof of Theorem~\ref{theo:strongly-connected}}

We start with the $\Pi_2^P$ upper bound. Let $\Gamma$ be a strongly connected UCRPQ. 
Let $\gamma(\bar x)$ be a disjunct of $\Gamma$. We define $\gamma^{<\infty}(\bar x)$ 
to be the CRPQ obtained from $\gamma(\bar x)$ by adding an atom $x\xrightarrow{\epsilon} x$, for each free variable $x$ in $\bar x$, 
and removing all atoms $y\xrightarrow{L} z$ such that $L$ is infinite. 
Note that $\gamma^{<\infty}(\bar x)$ could be not well-defined. (This happens precisely when $\gamma$ is Boolean and all of its RPQs are infinite.)
We define $\Gamma^{<\infty}:=\bigvee\{\gamma^{<\infty}: \text{$\gamma$ in $\Gamma$, and $\gamma^{<\infty}$ is well-defined}\}$. 
If every $\gamma^{<\infty}$ is not well-defined, then $\Gamma^{<\infty}$ is not well-defined neither. 
We say that a UCRPQ $\Gamma$ is \emph{$\epsilon$-trivial} if it has at most one free variable, and 
there is a CRPQ $\gamma$ in $\Gamma$ such that all of its RPQs contain the empty word $\epsilon$. 
Note that an $\epsilon$-trivial UCRPQ is always bounded. 
We have the following:

\begin{lemma}
\label{lemma:char-connected}
A strongly connected UCRPQ $\Gamma$ is bounded iff $\Gamma$ is $\epsilon$-trivial or,  $\Gamma^{<\infty}$ is well-defined and $\Gamma^{<\infty}\subseteq \Gamma$. 
\end{lemma}

\begin{proof}
From right to left, if $\Gamma$ is $\epsilon$-trivial then it is bounded. Otherwise, $\Gamma^{<\infty}\subseteq \Gamma$ and then $\Gamma^{<\infty}$ is equivalent to $\Gamma$ (as $\Gamma\subseteq \Gamma^{<\infty}$ always holds). Since $\Gamma^{<\infty}$ is bounded, then $\Gamma$ is also bounded. 
From left to right, suppose $\Gamma$ is bounded and not $\epsilon$-trivial. 
By Proposition~\ref{prop:basic}, there is $k\geq 1$ such that for every expansion $\lambda$ of $\Gamma$ 
there is an expansion $\lambda'$ of $\Gamma$ such that $\|\lambda'\|\leq k$ and $\lambda'\to\lambda$ \, ($\dagger$). 
We show first that  $\Gamma^{<\infty}$ is well-defined. 
By contradiction suppose this is not the case. In particular, all the RPQs in $\Gamma$ are infinite. 
We pick an arbitrary disjunct $\gamma$ of $\Gamma$ and 
an expansion $\lambda_{>k}$ of $\gamma$ obtained from choosing a word $w\in L$ with $|w|>k$, for every atom $x\xrightarrow{L} y$ of $\gamma$. 
By ($\dagger$), there is an expansion $\lambda'$ such that $\|\lambda'\|\leq k$ and $\lambda'\to \lambda_{>k}$. 
Since $\Gamma$ is not $\epsilon$-trivial, it follows that there at least one (non-equality) atom $x\xrightarrow{a} y$ in $\lambda'$.  
Since $\gamma$ is strongly connected, $\lambda'$ has a (labeled) directed cycle containing $x\xrightarrow{a} y$ (\ie, number of edges) at most $k$. 
Since every directed cycle in $\lambda_{>k}$ has length greater than $k$, we have a contradiction with the fact that $\lambda'\to \lambda_{>k}$. 

Now we show $\Gamma^{<\infty}\subseteq \Gamma$ using Lemma~\ref{lemma:cont-crpq}.
Let $\lambda$ be any expansion of $\gamma^{<\infty}$ in $\Gamma^{<\infty}$. If $\gamma=\gamma^{<\infty}$, then we are done. 
Otherwise, consider the expansion $\lambda_{>k}$ of $\gamma$ obtained by (1) choosing the same word as in $\lambda$ for atoms 
$x\xrightarrow{L} y$ with $L$ finite, and (1) choosing a word $w\in L$ such that $|w|>k$ for the atoms $x\xrightarrow{L} y$ with $L$ infinite.  
Note that we can partition  the (non-equality) atoms of $\lambda_{>k}$ into those generated in case (1), denoted by $A^{<\infty}$ and those generated in case (2), denoted by $A^{\infty}$. 
By ($\dagger$), there is an expansion $\lambda'$ of $\Gamma$ with $\|\lambda'\|\leq k$ such that $\lambda'\to\lambda_{>k}$ via a homomorphism $h$. 
We claim that the image via $h$ of every atom $x\xrightarrow{a} y$ in $\lambda'$ belongs to $A^{<\infty}$. By contradiction, 
suppose $h(x)\xrightarrow{a} h(y) \in A^{\infty}$. 
Since $\gamma$ is strongly connected, $\lambda'$ has a (labeled) directed cycle containing $x\xrightarrow{a} y$ of length $\leq k$. 
This is a contradiction as $\lambda'\to\lambda_{>k}$ and every directed cycle in $\lambda_{>k}$ has length $>k$. 
Hence, $\lambda'\to \lambda$. By Lemma~\ref{lemma:cont-crpq}, we obtain that $\Gamma^{<\infty}\subseteq \Gamma$. 
\end{proof}

For the lower bound, we reduce from the following well-known $\Pi_2^P$-complete problem: 
Given a connected (undirected) graph $G=(V,E)$ and $k\geq 1$ (given in unary), check whether for every function $c:V\to \{0,1\}$, there is a 
clique $K$ in $G$ of size $k$ such that $c(u)=c(v)$ for all nodes $u,v$ in $K$.
(Recall that a clique is a graph with an edge between each pair of distinct nodes.)

Given $G=(V,E)$ and $k\geq 1$, we define a Boolean strongly connected UCRPQ $\Gamma$ over the alphabet $\A:=\{a,b,0,1\}$ as follows. 
Let $\gamma_G$ be the Boolean CRPQ with variable set $\{x_u: u\in V\}$ where we have atoms $x_u\xrightarrow{a}x_v, x_u\xrightarrow{b}x_v, x_v\xrightarrow{a}x_u, x_v\xrightarrow{b}x_u$, for each edge $\{u,v\}\in E$, 
and an atom $x_u\xrightarrow{0+1}x_u$, for each node $u\in V$. 
For $\ell\in \{0,1\}$, we define the Boolean CRPQ $\gamma_k^\ell$ to have variable set $\{z_0,\dots,z_{k-1}\}$, atoms $z_i\xrightarrow{a}z_j, z_j\xrightarrow{a}z_i$, 
for each pair $i\neq j\in \{0,\dots,k-1\}$, and an atom $z_i\xrightarrow{\ell}z_i$, for each $i\in \{0,\dots,k-1\}$.
We pick an arbitrary node $u_0$ from $G$ and for every $\ell\in \{0,1\}$, 
we define $\gamma_G^\ell$ to be the Boolean CRPQ obtained from the (disjoint) conjunction of $\gamma_G$ and $\gamma_k^\ell$ by 
adding the atoms $x_{u_0}\xrightarrow{b^*}z_0, z_0\xrightarrow{b^*}x_{u_0}$. 
Then we let $\Gamma:=\gamma_G^0 \lor \gamma_G^1$. Note that $\Gamma$ is actually strongly connected. 
Also observe that $\Gamma^{<\infty}=(\gamma_G^0)^{<\infty} \lor (\gamma_G^1)^{<\infty}$, 
where $(\gamma_G^\ell)^{<\infty}$ is the (disjoint) conjunction of $\gamma_G$ and $\gamma_k^\ell$, for $\ell\in \{0,1\}$. 

We claim that $G,k$ is a positive instance iff $\Gamma$ is bounded. 
We show that $G,k$ is a positive instance iff $\Gamma^{<\infty}\subseteq \Gamma$, a hence the claim follows from Lemma~\ref{lemma:char-connected}. 
Suppose that $G,k$ is a positive instance. We prove $\Gamma^{<\infty}\subseteq \Gamma$ using Lemma~\ref{lemma:cont-crpq}. 
Let $\lambda$ be an any expansion of $\Gamma^{<\infty}$. 
Then there is $\ell\in\{0,1\}$ and a function $c:V\to\{0,1\}$ such that $\lambda$ is the disjoint conjunction of $\gamma_G^c$ and $\gamma_k^\ell$, 
where $\gamma_G^c$ is obtained from $\gamma_G$ by replacing $x_u\xrightarrow{0+1}x_u$ with $x_u\xrightarrow{c(u)}x_u$, for each $u\in V$. 
By hypothesis, $G$ contains a clique $K$ of size $k$ with $c(z)=\ell'$, for each $z$ in $K$ and some fixed $\ell'\in \{0,1\}$. Pick an arbitrary node $z^*$ of $K$ and 
consider the expansion $\lambda'$ of $\gamma_G^{\ell'}$ given by $\lambda_G^c$ $\land \gamma_k^{\ell'}$ $\land x_{u_0}\xrightarrow{b^d}z_0$ $\land z_0\xrightarrow{b^d}x_{u_0}$, 
where $d\geq 0$ is the distance in $G$ from $u_0$ to $z^*$. 
Then $\lambda'\to \lambda_G^c\to \lambda$ via the homomorphism $h$ that is the identity over $\lambda_G^c$ 
and maps $\gamma_k^{\ell'}$ to $\{x_z:\text{$z$ in $K$}\}$ with $h(z_0)=x_{z^*}$. Hence, $\Gamma^{<\infty}\subseteq \Gamma$. 

Suppose now that $\Gamma^{<\infty}\subseteq \Gamma$ and let $c$ be any function $c:V\to\{0,1\}$. 
Consider the expansion $\lambda_c$ of $\gamma_G^{0}$ given by $\gamma_G^c\land \gamma_k^0$, where $\gamma_G^c$ 
is defined as above. By Lemma~\ref{lemma:cont-crpq}, there is an expansion $\lambda_\ell$ of some $\gamma_G^\ell$ with $\ell\in \{0,1\}$, 
such that $\lambda_\ell\to \lambda_c$. Since $\lambda_\ell$ is connected, and contains symbols $b\in \A$, 
it is the case that $\lambda_\ell\to \gamma_G^c$ via $h$. 
Hence $\{h(z_i):i\in \{0,\dots,k-1\}\}$ must correspond to a clique $K$ of $G$ with $c(z)=\ell$, for all $z$ in $K$.

\end{document}